\documentclass[11pt,a4paper]{amsart}
\usepackage{amssymb}
\usepackage{amsmath}
\usepackage{enumerate}
\usepackage{mathtools}
\usepackage{xcolor}
\usepackage{textcomp}
\usepackage{amsthm}
\usepackage{array}
\usepackage{bbm}
\usepackage[margin=2.35cm]{geometry}
\usepackage{float}
\usepackage{hyperref}
\usepackage{rotating}

\newcolumntype{L}{>{$}l<{$}}
\newcommand{\larry}{\langle}
\newcommand{\roger}{\rangle}
\renewcommand{\t}{\mathfrak{t}}
\newcommand{\F}{\mathbb{F}}
\newcommand{\g}{\mathfrak{g}}
\newcommand{\s}{\mathfrak{s}}

\newcommand{\T}{\mathcal{T}}
\newcommand{\D}{\mathbf{D}}

 \theoremstyle{plain}
\theoremstyle{definition}

\newtheorem{definition}{Definition}[section]
\newtheorem{hyp}{Hypotheses}[section]

\theoremstyle{remark}
\newtheorem{remark}[definition]{Remark}

\newtheorem{corollary}[definition]{Corollary}

\newtheorem{lemma}[definition]{Lemma}

\newtheorem{theorem}[definition]{Theorem}
\newtheorem{proposition}[definition]{Proposition}

\newcounter{claim}[definition]

\theoremstyle{remark}
\newtheorem{example}[definition]{Example}

\theoremstyle{remark}

\theoremstyle{definition}

\DeclareMathOperator{\Aut}{Aut}

\DeclareMathOperator{\PSL}{PSL}

\DeclareMathOperator{\SL}{SL}

\DeclareMathOperator{\Sym}{Sym}

\DeclareMathOperator{\V}{\mathbb{V}}

\renewcommand{\epsilon}{\varepsilon}
\renewcommand{\hat}{\widehat}

\title[Digraph codes]{Constructing linear codes\\ from digraphs and groups}

\author{Coen del Valle}\address{School of Mathematics and Statistics,
The Open University, UK}
\email{Coen.del-Valle@open.ac.uk}
\author{Cheryl E. Praeger}\address{Center for the Mathematics of Symmetry and Computation, University of Western
Australia, Australia}\email{Cheryl.Praeger@uwa.edu.au}
\date{\today}
\thanks{This work was initiated during a visit by CdV to the University of Western Australia supported by a Cheryl~E.~Praeger Visiting Research Fellowship. CdV also acknowledges that his work is supported by the Engineering and Physical Sciences Research Council (EPSRC)
[grant number EP/Z534742/1].  CP would like to thank the Isaac Newton Institute for Mathematical Sciences for the support and hospitality when some work on this paper was undertaken during the programme  \emph{Algebraic groups, geometry, invariants and related topics (GGI)} (supported by EPSRC grant no. EP/Z000580/1). 
CP also acknowledges that her work was partially supported by a grant from the Simons Foundation.}

\begin{document}
\begin{abstract}
    In 2012, Kaufman and Lubotzky constructed the first family of symmetric LDPC good codes. Their construction used \emph{Cayley codes}, as originally defined by Kaufman and Wigderson (2016). In this paper we present two generalisations to the Cayley code construction, which we call \emph{graph codes} and \emph{digraph codes}. We investigate both the algebraic, and combinatorial properties of these constructions and show that they possess the same desirable attributes as Cayley codes, but with added freedom. We analyse the relationship between the expansion properties of the ingredient (di)graphs and the parameters of the constructed codes; our analysis offers an improvement to the results of Kaufman and Lubotzky. As an application, we construct an infinite family of good digraph codes, and we propose a series of open problems.
\end{abstract}
\maketitle

{\begin{center}
    {\tiny\emph{Dedicated to Alex Lubotzky on the occasion of his 70th birthday with our admiration\\ and thanks for his beautiful contributions to groups, graphs, and codes.}}
\end{center}}

\section{Introduction}
An $[n,k,d]_q$-code is a dimension-$k$ subspace $C\leq \V:=\mathrm{GF}(q)^\Omega$ where $|\Omega|= n$ and each non-zero element of $C$ has (Hamming) weight at least $d$. A family of codes is called \emph{good} if both the rate $r(C):=k/n$ and relative distance $\delta(C):=d/n$ are bounded away from 0, and is \emph{low-density parity-check (LDPC)} if $C$ can be defined from constraints of bounded weight. The action of $\Sym(\Omega)$ on $\Omega$ naturally induces an action on $\mathbb{V}$; when the subgroup of $\Sym(\Omega)$ preserving $C$ is transitive (on $\Omega$) we call $C$ \emph{symmetric}. (These definitions will be restated rigorously in Section~\ref{sec:prelim}.)

 Given a group $G$, a subset $S\subseteq G\setminus\{1\}$ which is inverse-closed (that is, $S=\{\s^{-1} \mid \s \in S\}$), and a linear code $B$ `compatible' with $S$, Kaufman and Wigderson~\cite{KauWig} introduced the \emph{Cayley code} $\mathbf{C}(G,S,B)$. In 2012, Kaufman and Lubotzky~\cite{KaufLub} exploited some of the many nice properties of these codes to construct the first family of symmetric LDPC good codes, and a recent paper~\cite{AruPraRad} examined in more generality the properties of the class of Cayley codes. The desirable properties of Cayley graphs which make this construction possible (and so powerful) are also possessed by the wider class of \emph{vertex-transitive} (di)graphs. As such, we introduce two new constructions of linear codes based on this observation.

 The first construction is the \emph{graph code} (Definition~\ref{d:gcode}). Fix a finite field $\F$, a (finite, undirected) graph $\Gamma=(V,E)$ with a group $G$ acting transitively on its vertices, and a vertex $\alpha\in V$. Let $\mathcal{T}$ be a transversal (that is, a set of right coset representatives) for the point stabiliser $G_\alpha$ in $G$, and let $B\leq \F^{\Gamma(\alpha)}$ be a linear  code defined on the neighbourhood $\Gamma(\alpha)$ of $\alpha$ in $\Gamma$. In Section~\ref{sec:basicprops}, we define the graph code $\mathbf{G}_\Gamma(G,\T,\alpha,B)$, which specialises to the Cayley code in the case that the action of $G$ on $V$ is regular (see Proposition~\ref{prop:GtoC}).

 Our second construction is the \emph{digraph code} (Definition~\ref{def:digraph}), which is the natural directed analogue of the graph code. The digraph code construction takes as input a (finite) directed graph $\Gamma=(V,A)$ with a group $G$ acting transitively on $V$, as well as a vertex $\alpha\in V$, transversal $\T$ for $G_\alpha$ in $G$, and linear codes $B_1\leq \F^{\Gamma^+(\alpha)}$ and $B_2\leq\F^{\Gamma^-(\alpha)}$ defined on the out- and in-neighbourhoods of $\alpha$, respectively.

 Our first main result shows that, if we choose the input codes to be `compatible' with the action of $G$, then the output code will also be compatible with $G$---a stronger statement appears as Theorem~\ref{thm:AutomorphismsU}. To conserve space we shall state the result for the graph code, and flag Theorem~\ref{thm:AutomorphismsD}, which is the analogous result for digraph codes. Given a group $G$ acting on a set $\Omega$ and a subset $\Lambda\subseteq \Omega$ left invariant by $G$, we use $G^\Lambda$ to denote the permutation group induced by $G$ on $\Lambda$. Similarly, a group $G$ of automorphisms of a graph $(V,E)$ induces a group of permutations of $E$, denoted by $G^E\leq\mathrm{Sym}(E)$.

\begin{theorem}\label{thm:AutomorphismsUmain}
    Fix a finite field $\F$, a (finite, undirected) graph $\Gamma=(V,E)$ with a group $G$ of automorphisms acting transitively on $V$, and $\alpha\in V$. Let $\mathcal{T}$ be a transversal for $G_\alpha$ in $G$ with $1\in\T$, and let $B\leq \F^{\Gamma(\alpha)}$, a linear code. Set $C=\mathbf{G}_\Gamma(G,\T,\alpha,B)$. If $G_\alpha^{\Gamma(\alpha)}\leq \Aut(B)$, then $G^{E}\leq \Aut(C)$.
\end{theorem}

 There are two key consequences of Theorem~\ref{thm:AutomorphismsUmain}. First, if we choose $B$ such that $G_\alpha^{\Gamma(\alpha)}\leq\Aut(B)$ then the graph code is invariant of the choice of $\T$ (Corollary~\ref{cor:Tindep}). Secondly, Theorem~\ref{thm:AutomorphismsUmain} allows us to construct symmetric LDPC codes, and this will be our second main result. Again, the coming theorem is a weak form of Corollary~\ref{cor:symmetry}, which also includes an analogous result for digraph codes.

 \begin{theorem}\label{thm:ldpc}
      Fix a finite field $\F$, a (finite, undirected) graph $\Gamma=(V,E)$ with a group $G$ of automorphisms acting transitively on $V$, and $\alpha\in V$. Let $\mathcal{T}$ be a transversal for $G_\alpha$ in $G$ with $1\in\T$, and let $B\leq \F^{\Gamma(\alpha)}$, a linear code. Set $C=\mathbf{G}_\Gamma(G,\T,\alpha,B)$. If $G_\alpha^{\Gamma(\alpha)}$ is a transitive subgroup of $\Aut(B)$ then $C$ is symmetric, and if $B^\perp$ is generated by a single orbit of $G_\alpha^{\Gamma(\alpha)}$, then $C^\perp$ is generated by the orbit $c^G$ for some $c\in C$ of weight at most $|\Gamma(\alpha)|$
 \end{theorem}

Our third main result allows us to restrict our attention to connected graphs.

\begin{theorem}\label{thm:conndecompmain}
      Fix a finite field $\F$. Let $\Gamma=(V,E)$ be a (finite, undirected) graph with a group $G$ of automorphisms acting transitively on $V$. Fix $\alpha\in V$ and let $\Sigma$ be the connected component of $\Gamma$ which contains $\alpha$. Let $\mathcal{T}$ be a transversal for $G_\alpha$ in the setwise stabiliser $G_{\{\Sigma\}}$ with $1\in\mathcal{T}$ and let $\mathcal{S}$ be a transversal for $G_{\{\Sigma\}}$ in $G$ with $1\in\mathcal{S}$. Then for any linear code $B\leq \F^{\Gamma(\alpha)}$ we have a direct sum decomposition \begin{equation}\label{eq:disconnectU}
      \mathbf{G}_\Gamma(G,\mathcal{TS},\alpha,B)=\bigoplus_{\s\in \mathcal{S}} \mathbf{G}_{\Sigma^\s}(G_{\{\Sigma^\s\}},\mathcal{T}^\s,\alpha^\s,B^\s),
      \end{equation}
    where  all direct factors are isomorphic.
\end{theorem}

As usual, we have an analogous result for digraph codes, namely Theorem~\ref{thm:decompD}. One could alternatively want codes which admit a direct sum decomposition into symmetric direct summands; this can be achieved from the digraph and graph code constructions, as shown by Theorems~\ref{thm:orbitaldecompD} and~\ref{thm:orbitaldecompU}.

With a view to constructing infinite families of good codes, we provide lower bounds on both the rate and relative distances of graph codes and digraph codes. Our bound for graph codes offers a modest improvement to the standard result for Cayley codes, and hence we are able to slightly improve the result~\cite[Theorem 11]{KaufLub} (see Corollary~\ref{cor:KauLuBound}). The following is a compressed version of Theorems~\ref{thm:distanceboundU} and~\ref{thm:ratebound}; the analogous statements for digraph codes are Theorems~\ref{thm:distanceboundD} and~\ref{thm:ratebound}.

\begin{theorem}\label{thm:mainstats}
    Fix a finite field $\F$, a (finite, undirected) graph $\Gamma=(V,E)$ with a group $G$ of automorphisms acting transitively on $V$, and $\alpha\in V$. Let $\mathcal{T}$ be a transversal for $G_\alpha$ in $G$ with $1\in\T$, and let $B\leq \F^{\Gamma(\alpha)}$, a linear code. Set $C=\mathbf{G}_\Gamma(G,\T,\alpha,B)$. Define $v=|\Gamma(\alpha)|$, let $M$ be the adjacency matrix of $\Gamma$, and let $\lambda$ be the absolute value of the eigenvalue of $M$ of second greatest magnitude. Then $r(C)\geq 2r(B)-1$ and if $\lambda<v$ then $$\delta(C)\geq \delta(B)\cdot \left(\frac{\delta(B)-(\lambda/v)}{1-(\lambda/v)}\right).$$
\end{theorem}

See Section~\ref{sec:stats} for a more explicit description of the eigenvalue $\lambda$ used in Theorem~\ref{thm:mainstats} (note that $\lambda$ is always less than $v$ unless $\Gamma$ is bipartite).

The final section of this paper is dedicated to applying our constructions in specific scenarios; in Section~\ref{sec:bipartite} we show how our constructions generalise the direct (Kronecker) product construction in a natural way. And in Section~\ref{sec:goodcodes} we use Theorems~\ref{thm:distanceboundD} and~\ref{thm:ratebound} to construct a family of good digraph codes; we do not make an effort to insist upon any symmetry in these codes. With this and Theorem~\ref{thm:ldpc} in mind, we pose the following open problems which we would like to see resolved:

    \begin{enumerate}
        \item Does there exist an infinite family of symmetric LDPC good graph codes which are not constructed from Cayley graphs?
        \item Does there exist an infinite family of symmetric LDPC good `proper' digraph codes?
        \item Does there exist an infinite family of symmetric LDPC good digraph codes such that the input codes $B_1$ and $B_2$ are non-isomorphic?
    \end{enumerate}
By proper digraph we mean that there is some arc $(x,y)$ where $(y,x)$ is not an arc. Since we are mainly concerned with vertex-transitive digraphs---or even edge-transitive digraphs---our digraphs will often be the orbital digraph corresponding to a non-self paired orbital, and hence will be proper digraphs.

The structure of the paper is straightforward. In Section~\ref{sec:prelim} we explain the necessary background on group theory, graph theory, and coding theory. In Section~\ref{sec:basicprops} we introduce the graph code and digraph code constructions, and explore some of their basic properties and prove Theorems~\ref{thm:AutomorphismsUmain} and~\ref{thm:ldpc}. Next, in Section~\ref{sec:isocondec} we investigate how the output codes are impacted by graph theoretic properties such as graph isomorphism and connectivity, before proving bounds on rate and relative distance in Section~\ref{sec:stats}. Finally, we wrap up the paper by working through specific examples and constructions in Section~\ref{sec:examples}.

\section{Preliminaries}\label{sec:prelim}
As this paper is cross-disciplinary we shall begin by recording some key definitions and notation which we shall use regularly later on.
\subsection{Group theory definitions and notation}
Throughout, we will only be interested in finite groups. These groups will be denoted by capital Latin symbols, usually $G$, with their elements denoted using Fraktur symbols (e.g. $\g,\t$). Let $G$ be a finite group.

    An \emph{action} of $G$ on a set $\Omega$ is a homomorphism $\rho : G\to\mathrm{Sym}(\Omega)$. Given an element $\alpha\in \Omega$ we write $\alpha^\g$ to mean `the point of $\Omega$ which is the 
    image of $\alpha$ under the permutation $(\g)\rho$'. The \emph{orbit} of $\alpha$ is the set $\alpha^G=\{\alpha^\g : \g \in G\}$, and the action of $G$ is \emph{transitive} if $\alpha^G=\Omega$.

    The \emph{stabiliser} of a point $\alpha\in \Omega$ is the subgroup $G_\alpha:=\{\g\in G \mid \alpha^\g=\alpha\}$. For a subset $\Sigma\subseteq\Omega$, the \emph{pointwise stabiliser} of $\Sigma$ is $$G_{(\Sigma)}=\{\g\in G\mid \beta^\g=\beta \text{ for all $\beta\in\Sigma$}\}=\bigcap_{\beta\in\Sigma}G_\beta,$$ and the \emph{setwise stabiliser} of $\Sigma$ is $$G_{\{\Sigma\}}=\{\g\in G\mid \beta^\g\in\Sigma \text{ for all $\beta\in \Sigma$}\}.$$
    The setwise stabiliser $G_{\{\Sigma\}}$ acts naturally on $\Sigma$, and the kernel of this action is the pointwise stabiliser $G_{(\Sigma)}$; thus, there is a group of permutations of $\Sigma$ induced by $G$ denoted $G^{\Sigma}:=G_{\{\Sigma\}}/G_{(\Sigma)}\leq \Sym(\Sigma).$
    
    A transitive action of $G$ is \emph{regular} if $|G|=|\Omega|$, or equivalently, the stabiliser of any point is trivial.

Observe that an action of $G$ on a set $\Omega$ naturally induces an action of $G$ on the set of all fixed-length tuples of points of $\Omega$ and on the set of all fixed-size sets of points of $\Omega$. Moreover, if $G$ acts on $\Omega$ and $H\leq G$ then there is an induced action of $H$ on $\Omega$.

Let $H$ be a subgroup of $G$. A \emph{right transversal} for $H$ in $G$ is a set of representatives for the right cosets of $H$ in $G$. For convenience we will usually write just `transversal'---this should always be taken to mean a `right transversal'.
\subsection{Graph theory definitions and notation}
As for groups, we are only interested in finite graphs and digraphs (directed graphs). We shall typically denote both graphs and digraphs using $\Gamma$, with vertices denoted by small Greek symbols. We use $V\Gamma$ to denote the vertex set of $\Gamma$ and $E\Gamma$ for the edge set when $\Gamma$ is undirected, and $A\Gamma$ for the arc set when $\Gamma$ is directed. Throughout, we shall assume that our (di)graphs are simple, that is, our (di)graphs will have no loops, and there will be no repeated edges (arcs).

Given an undirected graph $\Gamma$, and $\alpha\in V\Gamma$ we use $\Gamma(\alpha)$ to denote the \emph{neighbourhood} of $\alpha$ in $\Gamma$, that is, the subset of $V\Gamma$ consisting of vertices which are in an edge with $\alpha$. The \emph{valency} of $\alpha$ is $|\Gamma(\alpha)|$. Similarly, if $\Gamma$ is directed we write $\Gamma^+(\alpha):=\{\beta\in V\Gamma\mid (\alpha,\beta)\in A\Gamma\}$ the \emph{out-neighbourhood} of $\alpha$, and $\Gamma^-(\alpha):=\{\beta\in V\Gamma\mid (\beta,\alpha)\in A\Gamma\}$ the \emph{in-neighbourhood} of $\alpha$. The out-valency and in-valency of $\alpha$ are $|\Gamma^+(\alpha)|$ and $|\Gamma^-(\alpha)|$, respectively.

Both directed and undirected graphs $\Gamma$ admit groups of automorphisms, denoted $\Aut(\Gamma)$, namely the subgroup of $\mathrm{Sym}(V\Gamma)$ consisting of all permutations leaving invariant $A\Gamma$ or $E\Gamma$, respectively. The group $\Aut(\Gamma)$ and its subgroups thus induce actions on $V\Gamma$ and $E\Gamma$ (or $A\Gamma$, in the directed case). Let $G\leq \Aut(\Gamma)$. We say that $G$ is \emph{vertex-transitive} if the action of $G$ on $V\Gamma$ is transitive, and \emph{edge-transitive} (resp. \emph{arc-transitive}) if the action of $G$ on $E\Gamma$ (resp. $A\Gamma$) is transitive. The (di)graph $\Gamma$ is \emph{vertex-}, \emph{edge-}, or \emph{arc-transitive} if $\Aut(\Gamma)$ is. 

Observe that when an undirected graph is vertex-transitive, then all vertices necessarily have the same valency. Similarly, should a directed graph be vertex-transitive then all out-valencies will be equal, as will all in-valencies, and moreover, since our digraphs are assumed to be finite, these two quantities must be the same. Therefore, in the vertex-transitive case we may refer to the \emph{valency} of $\Gamma$ to mean `the valency of any vertex of $\Gamma$', and similarly for in- and out-valencies.

\begin{definition}
    Let $G$ be a finite group and $S\subseteq G\setminus\{1\}$. If $S$ is inverse-closed (i.e. $S=\{\s^{-1} \mid \s \in S\}$) then we define the \emph{Cayley graph}, $\mathrm{Cay}(G,S)$ to be the graph $\Gamma$ with vertex set $V\Gamma=G$ and edge set $E\Gamma=\{\{\g,\s\g\} \mid \g\in G\}$. Should $S$ not be inverse-closed, then we can define the \emph{Cayley digraph} $\mathrm{Cay}(G,S)$ to be the directed graph with vertex set $G$ and arc set $A\Gamma=\{(\g,\s\g) \mid \g\in G\}$.
\end{definition}
It is not hard to see that Cayley graphs are vertex-transitive. In fact, the natural action of $G$ by right multiplication on $\mathrm{Cay}(G,S)$ is regular.
\subsection{Coding theory definitions and notation}\label{sec:codedefs}
Let $\F=\mathrm{GF}(q)$ be a finite field of order $q$. A linear code over $\F$ is a subspace $C$ of the vector space $\V:=\F^\Omega$ of dimension $n=|\Omega|$. We generally view $\V$ as the set of all functions $f:\Omega\to \F$ with pointwise addition and multiplication, that is,  
\[
\lambda f+\lambda'f':\alpha\mapsto \lambda(\alpha)f + \lambda'(\alpha)f'\quad \text{for $f,f'\in \mathbb{V}, \lambda,\lambda'\in \F$, and $\alpha\in\Omega$.}
\]
(We apply our functions on the \emph{right} rather than the oft-standard left convention used in coding theory; that is $(\alpha)f$ means `the image of $\alpha$ under $f$'.)

The dimension $n(C)$ of the ambient space $\V$ is the \emph{length} of the code, and the dimension $k(C)$ of $C$ (as a vector space) is the \emph{dimension} (or \emph{rank}) of the code $C$; the \emph{rate} of $C$ is $r(C)=k/n$. Elements of $C$ are called \emph{codewords}, and the \emph{(Hamming) weight} of $c\in C$ is the number $\mathrm{wt}(c)=|\{\alpha\in \Omega \mid (\alpha)c\ne 0\}|$ of non-zero entries of $c$. The \emph{minimum distance} of $C$ is $d(C)=\min_{c\in C\setminus\{0\}}\mathrm{wt}(c)$; the \emph{relative distance} (sometimes called the `normalised distance') of $C$ is $\delta(C)=r(C)/n$. 

Observe that $\Sym(\Omega)$ induces a natural action on $\V$ given by $f^\g:\alpha\mapsto(\alpha^\g)f$ for all $f\in V$, $\g\in \Sym(\Omega)$, and $\alpha\in \Omega$. The \emph{automorphism group} of $C$ is the subgroup $$\Aut(C)=\{\g\in G\mid c^\g\in C\text{ for all $c\in C$}\}\leq\Sym(\Omega)$$ of all permutations of $\Omega$ leaving $C$ invariant. The code $C$ is \emph{symmetric} if $\Aut(C)$ is a transitive subgroup of $\Sym(\Omega)$, and is \emph{symmetric with respect to $G$} if $G$ is a subgroup of $\Aut(C)$ which is transitive on $\Omega$. Two codes are \emph{isomorphic} (or \emph{equivalent}) if they are in the same $\Sym(\Omega)$-orbit.

Given $f,g\in \V$ we have an inner product between $f$ and $g$ given by 
$$
\sum_{\alpha\in \Omega}\left((\alpha)f\cdot(\alpha)g\right).
$$ 
With respect to this inner product, the orthogonal space to $C$ is the \emph{dual-code} $C^\perp$. The automorphism group $\Aut(C)$ of $C$ leaves $C^\perp$ invariant and conversely; the code $C$ is \emph{single-orbit symmetric with respect to $G$} if $C$ is symmetric with respect to $G$ and there exists some $f\in C^\perp$ such that the orbit $f^G$ spans $C^\perp$, that is $\langle f^\g : \g \in G\rangle=C^\perp.$

An $[n,k,d]_q$\emph{-code} is a linear code of length $n$, dimension $k$, and minimum distance $d$, defined over a field of order $q$. 

\begin{definition}\label{def:prodandsum}
 Given finite sets $\Omega$ and $\Sigma$,  $f\in\F^{\Omega\times\Sigma}$, and $\alpha\in\Omega$ we define $f_\alpha:\Sigma\to\F$ by $\beta\mapsto (\alpha,\beta)f$. Given $\beta\in \Sigma$ fixed, we define $f_\beta:\Omega\to\F$ analogously. For codes $C_1\leq \F^\Omega$ and $C_2\leq \F^\Sigma$, the \emph{direct product} of $C_1$ and $C_2$ is the code $$C_1\otimes C_2:=\{f\in\F^{\Omega\times\Sigma}\mid f_\alpha\in C_2 \text{ for all $\alpha\in \Omega$ and } f_\beta\in C_1\text{ for all $\beta\in\Sigma$}\}\leq \F^{\Omega\times\Sigma}.$$ If $\Omega$ and $\Sigma$ are disjoint, we can define the \emph{direct sum} of $C_1$ and $C_2$ to be the code $$C_1\oplus C_2:=\{f\in \F^{\Omega\cup\Sigma} \mid f|_\Omega\in C_1\text{ and }f|_\Sigma\in C_2\}.$$
\end{definition}

Let $\mathcal{C}=\{C_m\}_{m\in\mathbb{N}}$ be a family of codes. The family $\mathcal{C}$ is \emph{good} if there exists $\epsilon>0$ such that $r(C_m)$ and $\delta(C_m)$ are both at least $\epsilon$ for all $m$. If all the dual-codes $C_m^\perp$ are generated by codewords of bounded weight, then we say that $\mathcal{C}$ is \emph{low-density parity-check (LDPC)}. The family $\mathcal{C}$ is \emph{highly symmetric} if for each $m$, $C_m$ is symmetric and $C_m^\perp$ is spanned by $f_m^G$ for some $f_m$ of bounded weight. Clearly, any highly symmetric family is LDPC.

Although not the focus of the paper, we now define \emph{Cayley codes} as they are the construction we shall generalise and it will be useful to point to this definition when comparing our concepts and results.

\begin{definition}\label{def:CC}
 Let $\F$ be a finite field, let $G$ be a group, let $S\subseteq G\setminus\{1\}$ be inverse-closed and let $B\leq \F^S$ be a linear code.. Let $\Gamma=(V,E):=\mathrm{Cay}(G,S)$, and for $\g \in G$ define $N_\g:=\{ \{\g,\s\g\}\mid \s\in S\}$ the set of edges incident to $\g$, and $\chi_\g : S\to N_\g$ given by $\s\mapsto\{\g,\s\g\}.$ The \emph{Cayley code} $\mathbf{C}(G,S,B)$ is the code
\begin{equation*}\label{e:cc}
    \mathbf{C}(G,S,B) =\{ f\in\F^E \mid \chi_\g\circ f\in B\ \text{for all $\g\in G$} \}.
\end{equation*}

\end{definition}

\section{Graph codes, digraph codes, and their basic properties}\label{sec:basicprops}
We start by defining graph codes, a natural generalisation of Cayley codes. For the remainder of the paper we shall fix a finite field $\F$.
\begin{definition}\label{d:gcode}
{\rm
    Let $\Gamma=(V, E)$ be an undirected graph admitting a vertex-transitive subgroup $G\leq \Aut(\Gamma)$. Let $\alpha\in V$ and $S=\Gamma(\alpha)$, the set of vertices adjacent to $\alpha$. Let $\mathcal{T}$ be a transversal for $G_\alpha$ in $G$ such that $1\in\mathcal{T}$, and for each $\t\in \mathcal{T}$ let $E_\t = \{ \{\alpha^\t,\beta^\t\}\mid \beta\in S\}$, and define
    \begin{equation*}\label{e:phit}
        \phi_\t:S\to E_\t\quad \text{by}\quad \beta\to\{\alpha^\t,\beta^\t\}\ \text{for $\beta\in S$.}
    \end{equation*}   
    For $B\leq \F^S$ a linear code in $\F^S$, define the \emph{graph code} $\mathbf{G}_\Gamma(\mathcal{T}, \alpha, B)$ by 
    \begin{equation*}\label{e:gcode}
       \mathbf{G}_\Gamma(G,\mathcal{T}, \alpha, B)=\{ f\in\F^E\mid  \phi_\t\circ f\in B\ \text{for all $\t\in \mathcal{T}$}\}\subseteq \F^E.
    \end{equation*}  
    }
\end{definition}
At some point we may be discussing more than one graph code (defined on different graphs) at the same time, at which point it will become convenient to write the maps $\phi_\t$ as $\phi_\t^\Gamma$ to make it clear which map we mean. 

We begin by showing that the graph code is indeed a linear code, and that the Cayley code construction is a specialisation of graph codes.

\begin{proposition}\label{prop:GtoC}
    Using the notation from Definition~$\ref{d:gcode}$, 
    \begin{enumerate}
        \item[(1)] For each $\t\in\T$, $E_\t$ is the set of edges incident with $\alpha^\t$, and $\phi_\t$ is a well-defined bijection $S\to E_\t$. Moreover, $C:=\mathbf{G}_\Gamma(G,\T,\alpha,B)$ is a linear code.
        \item[(2)] Suppose that $\Gamma=\mathrm{Cay}(G,S)$. Then $\mathbf{G}_\Gamma(G,\T, \alpha, B)$ is the Cayley code $\mathbf{C}(G,S,B)$.  
    \end{enumerate}
\end{proposition}
\begin{proof}
    The assertion about $E_\t$ follows since $E_1$ is the set of edges incident with $\alpha$ (by the definition of $S$), and since $\t\in\Aut(\Gamma)$. That $\phi_\t$ is a well-defined bijection now follows immediately. For the final claim let $\lambda\in \F$ and take $f,f'\in C$ and take $\t\in \T$. Then $$\phi_\t\circ(f+\lambda f')=(\phi_\t\circ f)+\lambda(\phi_\t\circ f)\in B,$$ hence (1).

    Suppose that $\Gamma=\mathrm{Cay}(G,S)=(V,E)$. Then $G$ acts regularly on $V$. Therefore, $G_\alpha=1$ and so $\T=G$. Thus, $\phi_\t=\chi_\t$ for each $\t\in \T$ by definition (where $\chi_\t$ is as defined in Definition~\ref{def:CC}). Since all constraints are the same and both codes are subspaces of $\F^E$ the codes must be equal, as asserted.
\end{proof}
We can generalise the construction further to directed graphs. We include analogous statements to Proposition~\ref{prop:GtoC}(1) as part of the coming definition since they are clear from identical arguments.
\begin{definition}\label{def:digraph}
    Let $\Gamma=(V,A)$ be a directed graph admitting a vertex-transitive subgroup $G\leq\Aut(\Gamma)$. Let $\alpha\in V$ and % and let $S^+=\Gamma^+(\alpha)$ and $S^-=\Gamma^-(\alpha)$ be the sets of out- and in-neighbours of $\alpha$, respectively. L
    let $\mathcal{T}$ be a transversal for $G_\alpha$ in $G$ with $1\in\mathcal{T}$. For each $\t\in\mathcal{T}$ let $O_\t=\{(\alpha^\t,\beta^\t) : \beta\in \Gamma^+(\alpha)\}$, and define $$\varphi_\t:\Gamma^+(\alpha)\to O_\t\text{ by } \beta\mapsto(\alpha^\t,\beta^\t).$$
    Similarly define $I_\t=\{(\beta^\t,\alpha^\t) : \beta\in \Gamma^-(\alpha)\}$, and define 
    $$
    \psi_\t:\Gamma^-(\alpha)\to I_\t\text{ by } \beta\mapsto(\beta^\t,\alpha^\t).
    $$
    The $O_\t$ and the $I_\t$ (for $t\in\mathcal{T}$) each form a partition of $A\Gamma$, and the maps $\varphi_\t$ and $\psi_\t$ are well-defined bijections. For linear codes $B_1\leq \F^{\Gamma^+(\alpha)}$ and $B_2\leq\F^{\Gamma^-(\alpha)}$, define the \emph{digraph code} 
    $$
    \mathbf{D}_\Gamma(G,\mathcal{T},\alpha,B_1,B_2)=\left\{f\in\F^{A\Gamma}\mid \varphi_\t\circ 
    f\in B_1 \text{ and }\psi_\t\circ f\in B_2 \text{ for all } \t\in \mathcal{T}\right\}
    \leq \F^{A\Gamma}.
    $$
\end{definition}
As with graph codes, for digraph codes we may sometimes write $\varphi_\t$ and $\psi_\t$ as $\varphi_\t^\Gamma$ and $\psi_\t^\Gamma$ to distinguish between these maps when multiple graphs are present. 

We briefly remark that one could also define codes from digraphs by only placing an input code on the out-neighbours, or by only placing an input code on the in-neighbours; the digraph code defined as above would then be the intersection of these. Unfortunately, these alternative constructions do not seem to yield very interesting codes; there simply are not enough constraints in place to guarantee nice structure. One can easily verify that such codes would have rate at least that of the input code, but obtaining some sort of useful general bound on the relative distance seems hopeless---experimentally it appears that the relative distance is generally extremely small.

Many of the coming results will use the same set-up, so for convenience we state the hypotheses once and reference them throughout.
\begin{hyp}\label{hyp:GorD}
    Let $\Gamma$ be a graph (directed or undirected). If $\Gamma$ is undirected (U) then let $G,\T,\alpha,$ and $B$ be as in Definition~\ref{d:gcode} and set $C=\mathbf{G}_\Gamma(G,\T,\alpha,B)$. If $\Gamma$ is directed (D) then let $G,\T,\alpha,B_1$ and $B_2$ be as in Definition~\ref{def:digraph} and set $C=\mathbf{D}_\Gamma(G,\T,\alpha,B_1,B_2)$.
\end{hyp}
Some of the upcoming results will have statements for both the directed and undirected cases within a single result; the statements marked by (U) and (D) are meant to have the undirected and directed versions of Hypotheses~\ref{hyp:GorD}, respectively.

%Observe that we can naturally define an action of $G$ on each of $\F^{E\Gamma}$ and on $\F^S$. Indeed, if $f\in\F^E$ we define $f^x : E\to\F$ by $e\mapsto (e^x)f$; the action on $\F^S$ is analogous. 
We begin by showing that by choosing our input codes to be `compatible' with the action of $G_\alpha$ on its neighbourhood in $\Gamma$, the output will always be `compatible' with the action of $G$ on the edges of $\Gamma$. 
\begin{theorem}\label{thm:AutomorphismsU}
    Assume Hypotheses~\ref{hyp:GorD} with $\Gamma$ undirected. % If $\Gamma$ is undirected take $C\in\{\mathbf{D}_\Gamma(\mathcal{T},\alpha,B),\mathbf{G}_\Gamma(\mathcal{T},\alpha,B)\}$, otherwise set $C=\mathbf{D}_\Gamma(\mathcal{T},\alpha,B)$. 
    Then the following hold: \begin{itemize}
    \item[(1)] the group $G$ induces an action on $\F^{E\Gamma}$: for $f\in\F^{E\Gamma}$ and $\g\in G$ we define $f^\g : E\Gamma\to\F$ by $e\mapsto (e^\g)f$;
    \item[(2)] the group $G_\alpha$ induces an action on $\F^{\Gamma(\alpha)}$: for $f\in\F^{\Gamma(\alpha)}$ and $\g\in G$ we define $f^\g : \Gamma(\alpha)\to\F$ by $\beta\mapsto (\beta^\g)f$;
    \item[(3)] if $G_\alpha^{\Gamma(\alpha)}\leq \Aut(B)$, then $G^{E\Gamma}\leq \Aut(C)$; and\item[(4)] if for each $b\in B$ there exists $c\in C$ and $\t\in \mathcal{T}$ such that $\phi_\t\circ c=b$ then the converse to (3) holds.\end{itemize}
\end{theorem}
The proof of Theorem~\ref{thm:AutomorphismsU} is omitted as it is extremely similar to the directed analogue proved below.
\begin{theorem}\label{thm:AutomorphismsD}
    Assume Hypotheses~\ref{hyp:GorD} with $\Gamma$ directed. % If $\Gamma$ is undirected take $C\in\{\mathbf{D}_\Gamma(\mathcal{T},\alpha,B),\mathbf{G}_\Gamma(\mathcal{T},\alpha,B)\}$, otherwise set $C=\mathbf{D}_\Gamma(\mathcal{T},\alpha,B)$. 
    Then the following hold: \begin{itemize}
    \item[(1)] the group $G$ induces an action on $\F^{A\Gamma}$: for $f\in\F^{A\Gamma}$ and $\g\in G$ we define $f^\g : A\Gamma\to\F$ by $e\mapsto (e^\g)f$;
    \item[(2)] the group $G_\alpha$ induces an action on $\F^{\Gamma^+(\alpha)}$: for $f\in\F^{\Gamma^+(\alpha)}$ and $\g\in G$ we define $f^\g : \Gamma^+(\alpha)\to\F$ by $\beta\mapsto (\beta^\g)f$. We get an analogous $G_\alpha$-action on $\Gamma^-(\alpha)$;
    \item[(3)] if $G_\alpha^{\Gamma^+(\alpha)}\leq \Aut(B_1)$ and $G_\alpha^{\Gamma^-(\alpha)}\leq \Aut(B_2)$ then $G^{A\Gamma}\leq \Aut(C)$; and
    \item[(4)] if for each $b_1\in B_1$ and $b_2\in B_2$ there exist $c_1,c_2\in C$ and $\t_1,\t_2\in \mathcal{T}$ such that $\varphi_{\t_1}\circ c_1=b_1$ and  $\psi_{\t_2}\circ c_2=b_2$ then the converse to (3) holds.\end{itemize}
\end{theorem}
\begin{proof}
    Statements (1) and (2) are easy to check.
    Suppose that $G_\alpha^{\Gamma^+(\alpha)}\leq \Aut(B_1)$ and $G_\alpha^{\Gamma^-(\alpha)}\leq \Aut(B_2)$. Let $\g\in G$ and $f\in C$; we aim to show that $f^\g\in C$. Let $\t\in \mathcal{T}$; by the definition of $\T$ there exist unique $\g'\in G_\alpha$, $\t'\in \mathcal{T}$ such that $\t\g=\g'\t'$. Since $\t'\in \T$ and $f\in C$, by definition we have $\varphi_{\t'}\circ f\in B_1$. %Moreover, since $G_\alpha^S\leq\Aut(B)$ we have that $(\varphi_{t'}\circ f)^{x'}\in B$. 
     Therefore, for every $\beta\in \Gamma^+(\alpha)$ we have, using the actions defined in both (1) and (2), \begin{align*}
        (\beta)(\varphi_\t\circ f^\g)=((\alpha^{\t\g},\beta^{\t\g}))f&=((\alpha^{\g'\t'},\beta^{\g'\t'}))f\\&=((\alpha^{\t'},\beta^{\g'\t'}))f\\
        &=(\beta^{\g'})(\varphi_{\t'}\circ f)\\
        &=(\beta)(\varphi_{\t'}\circ f)^{\g'}
    \end{align*}
    Since $\varphi_{\t'}\circ f\in B_1$ and $G_\alpha^{\Gamma^+(\alpha)}\leq \Aut(B_1)$, we deduce that $(\varphi_{\t'}\circ f)^{\g'}\in B_1$, and so $\varphi_\t\circ f^\g\in B_1$. We can similarly show that $\psi_\t\circ f^\g\in B_2$---since $\t$ was arbitrary it follows that $f^\g\in C$, proving (3).
    
    On the other hand, suppose that $G^{A\Gamma}\leq\Aut(C)$ and that $b\in B_1$ with $c\in C$ and $\t\in \mathcal{T}$ such that $\varphi_\t\circ c=b$. Let $\g\in G_\alpha$. Then $b^\g=(\varphi_\t\circ c)^\g$, and for each $\beta\in \Gamma^+(\alpha)$ we have \begin{align*}(\beta)b^\g=
        (\beta)(\varphi_\t\circ c)^\g=(\beta^\g)(\varphi_\t\circ c)=((\alpha^\t,\beta^{\g\t}))c&=((\alpha^{\g\t},\beta^{\g\t}))c\\&=((\alpha,\beta))c^{\g\t}\\&=(\beta)(\varphi_1\circ c^{\g\t})%\\&=(\{\alpha^{t'x'},\beta^{t'x'}\})c\\&=(\{\alpha^{t'},\beta^{t'}\})c^{x'}\\&=(\beta)(\varphi_{t'}\circ c^{x'}).\\
    \end{align*}
    By assumption $\g\t$ induces an automorphism of $C$, and hence $c^{\g\t}\in C$, whence $b^\g=(\varphi_1\circ c^{\g\t})\in B_1$ by the definition of $C$. Therefore $\g$ induces an automorphism of $B_1$. We can similarly check that $\g$ induces an automorphism of $B_2$, proving (4).
    %$f^x\in C$, so $(\varphi_{t'}\circ f)^{x'}=\varphi_t\circ f^x\in B$
\end{proof}

\begin{remark} {\rm
The additional condition of statement (4) of Theorems~\ref{thm:AutomorphismsU} and~\ref{thm:AutomorphismsD} is necessary. Indeed, consider the Kneser graph $\Gamma=K_{7:2}$ (that is the graph with vertices labelled by 2-element subsets of $\{1,\dots,7\}$ with $\{x,y\}\in E\Gamma$ if and only if $x\cap y=\emptyset$) and the $(10,3,5)$-code $B$ with generator matrix $$\begin{bmatrix}
1& 0& 0& 1 &0 &1& 1& 0& 0& 1\\
0& 1& 0& 1& 1& 0& 1& 0& 1& 0\\
0& 0& 1& 1& 1& 1& 1& 1& 0& 0
\end{bmatrix},$$ defined over $\F=\mathrm{GF}(2)$.
We verify using the computational algebra system {\sc Magma}~\cite{magma} that the automorphism group of $B$ is a 2-group, whereas choosing $\alpha:=\{1,2\}$ we have that $G_\alpha^{\Gamma(\alpha)}=S_5$ (equipped with its degree 10 action), and so $G_\alpha^{\Gamma(\alpha)}\not\leq\Aut(B)$. On the other hand, we compute that $\mathbf{G}(G,\mathcal{T},\alpha,B)$ is the zero code (so the projections $\phi_\t$ do not map onto any non-zero vector), hence $\Aut(C)=S_{105}> G^{E\Gamma}$.

%We next remark that the additional condition required in statement (2) is desirable in general, regardless of the theorem. Indeed, given $b\in B$, should there not be a pair $(c,t)\in C\times T$ such that $\varphi_t\circ c$
}
\end{remark}
It may seem a bit dissatisfying that we need the additional assumption in statement (4) of Theorems~\ref{thm:AutomorphismsU} and~\ref{thm:AutomorphismsD}, but in fact---as the 
next result shows---this assumption is no hindrance.

%In what follows, should the statement of a result not be affected by whether $\Gamma$ is directed or undirected and whether the constructed code is a graph code or a digraph code we shall abuse notation and write ``take $C=\mathbf{C}_\Gamma(G,\mathcal{T},\alpha,B_1,B_2)$" to mean ``set $C$ to be the code $\mathbf{D}_\Gamma(G,\mathcal{T},\alpha,B_1,B_2)$ if $\Gamma$ is directed, and take $C$ to be $\mathbf{G}_\Gamma(G,\mathcal{T},\alpha,B_1)$ otherwise". Should the $\mathbf{C}_\Gamma$ notation be used more than once in a statement it should be assumed that it is always referring to the same type of construction throughout.

\begin{proposition}\label{prop:Bbar}
    Assume Hypotheses~\ref{hyp:GorD}. If $\Gamma$ is undirected define $\overline{B}:=\{b\in B \mid \exists (c,\t)\in C\times\mathcal{T} \text{ such that } \phi_\t\circ c=b\}$ and if $\Gamma$ is directed define $\overline{B_1}:=\{b\in B_1 \mid \exists (c,\t)\in C\times\mathcal{T} \text{ such that } \varphi_\t\circ c=b\}$ and $\overline{B_2}:=\{b\in B_2 \mid \exists (c,\t)\in C\times\mathcal{T} \text{ such that } \psi_\t\circ c=b\}$. Then the following hold:
    \begin{itemize}
        \item[(U)] if $G^{E\Gamma}\leq \Aut(C)$ then $\overline{B}=\{\phi_1\circ c \mid c\in C\}$, and $\overline{B}$ is a linear code. Thus, $C=\mathbf{G}_\Gamma(G,\mathcal{T},\alpha,\overline{B})$ and $C$ satisfies the conditions of Theorem~\ref{thm:AutomorphismsU}(4) with respect to $\overline{B}$.
        \item[(D)]  if $G^{A\Gamma}\leq \Aut(C)$ then $\overline{B_1}=\{\varphi_1\circ c \mid c\in C\}$, $\overline{B_2}=\{\psi_1\circ c \mid c\in C\}$, and both $\overline{B_1}$ and $\overline{B_2}$ are linear codes. Thus, $C=\mathbf{D}_\Gamma(G,\mathcal{T},\alpha,\overline{B_1},\overline{B_2})$ and $C$ satisfies the conditions of Theorem~\ref{thm:AutomorphismsD}(4) with respect to $\overline{B_1}$ and $\overline{B_2}$.
    \end{itemize}
\end{proposition}
\begin{proof}
    We prove (U); the proof of (D) is similar. By the definitions of $\phi_\t$ (Definition~\ref{d:gcode}) and $c^\t$ (Theorem~\ref{thm:AutomorphismsU}(1)), $\phi_\t\circ c=\phi_1\circ c^\t$ for every $(c,\t)\in C\times\mathcal{T}$, so the first claim follows from the fact that $\t$ induces an automorphism of $C$. 

    For the second claim, let $b_1,b_2\in \overline{B}$. Then $b_1=\varphi_1\circ c_1$ and $b_2=\varphi_1\circ c_2$ for some $c_1,c_2\in C$. Thus, $b_1+\lambda b_2=\varphi_1\circ(c_1+\lambda c_2)\in \overline{B}$ for every $\lambda\in\F$, hence $\overline{B}$ is a linear code. The final claim is now immediate.
\end{proof}
\begin{corollary}
    Assume Hypotheses~\ref{hyp:GorD}. Then the following hold:
    \begin{itemize}
        \item[(U)] if $\Gamma$ is undirected then $G^{E\Gamma}\leq\Aut(C)$ if and only if $G_\alpha^{\Gamma(\alpha)}\leq \Aut(\overline{B});$ and
        \item[(D)] if $\Gamma$ is directed then $G^{A\Gamma}\leq\Aut(C)$ if and only if $G_\alpha^{\Gamma^+(\alpha)}\leq \Aut(\overline{B}_1)$ and $G_\alpha^{\Gamma^-(\alpha)}\leq \Aut(\overline{B}_2)$.
    \end{itemize}
\end{corollary}
Proposition~\ref{prop:Bbar} tells us that we lose no information by assuming that our choice of input code satisfies $B=\overline{B}$.

When studying the algebraic properties of these codes we will generally make the assumption $G_\alpha^{\Gamma(\alpha)}\leq \Aut(B)$; an easy corollary of Theorem~\ref{thm:AutomorphismsU} is that in this set-up the code $\mathbf{G}_\Gamma(G,\mathcal{T},\alpha,B)$ is independent of the choice of transversal $\mathcal{T}$.
\begin{corollary}\label{cor:Tindep}
    Assume Hypotheses~\ref{hyp:GorD} and let $\mathcal{T}'$ be a transversal of $G_\alpha$ in $G$ with $1\in\T'$. Let $C'$ be the code obtained by replacing $\mathcal{T}$ by $\mathcal{T}'$ in the definition of $C$. If $\Gamma$ is an undirected graph then assume that $G_\alpha^{\Gamma(\alpha)}\leq \Aut(B)$, and if $\Gamma$ is a digraph assume that $G_\alpha^{\Gamma^+(\alpha)}\leq \Aut(B_1)$ and $G_\alpha^{\Gamma^-(\alpha)}\leq \Aut(B_2)$. Then $C'=C$.
\end{corollary}
\begin{proof}
    We prove the result for $C$ a graph code, the digraph case is identical. Suppose that $c\in C$ and let $\t\in \mathcal{T}'$. Then $c^{\t}\in C$ by Theorem~\ref{thm:AutomorphismsU}(3). In particular, $\phi_\t\circ c=\phi_1\circ c^\t\in B$, so since $\t$ was arbitrary we deduce that $c\in C'$, whence $C\subseteq C'$. The other direction is identical, hence the result.
\end{proof}
The assumption that $G_\alpha^{\Gamma(\alpha)}\leq \Aut(B)$ is necessary, as illustrated by the following example
\begin{example}
    Let $\Gamma$ be the 3-prism with vertices labelled as in Figure~\ref{fig:3prism}. Set $G=\Aut(\Gamma)$, which is verified to be the dihedral group of order 12 generated by $(2\,3)(5\,6)$ and $(1\,5\,3\,4\,2\,6)$. Then $G_1^{\Gamma(1)}=\larry (2\,3)\roger$. Now, take $\F=\mathrm{GF}(4)$ and let $x$ be a primitive element of $\F$, and consider the $[3,1,2]_4$-code $B=\langle (1,0,x)\rangle\leq \F^{\Gamma(1)}$ (here we are fixing ascending order for the neighbourhood of $1$, and by $(1,0,x)$ we mean the function $\Gamma(1)\to \F$ mapping $2\mapsto 1$, $3\mapsto0$, and $4\mapsto x$). To keep our calculations compact we will abuse notation by writing $(i,j,k)b$ to refer to the tuple $((i)b,(j)b,(k)b)$ whenever $b\in B$.
    \begin{figure}[h]
        \centering
        \includegraphics[width=0.3\linewidth]{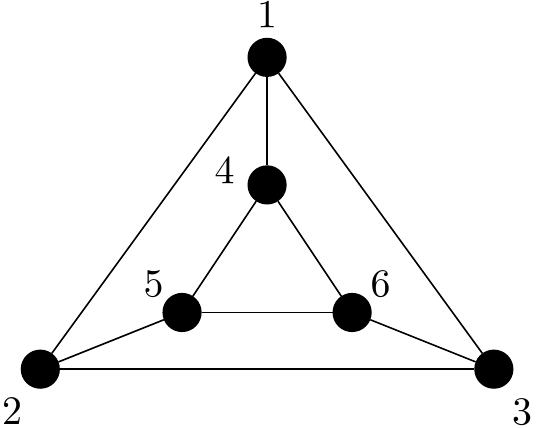}
        \caption{The 3-prism.}
        \label{fig:3prism}
    \end{figure}
    
    Notice that $\Aut(B)$ is trivial, and so $G_1^{\Gamma(1)}\not\leq \Aut(B)$. We now choose two transversals which we shall see yield distinct graph codes. Start by defining the partial transversal 
    $$
    \mathcal{T}_0:=\{1,(1\,2)(4\,5),(1\,3)(4\,6),(1\,5)(2\,4)(3\,6),(1\,6)(2\,5)(3\,4)\}.
    $$ 
    Both transversals $\mathcal{T}$ we construct will contain $\mathcal{T}_0$, so we begin by investigating the constraints imposed by $\mathcal{T}_0$. Suppose that $f\in\F^{E\Gamma}$ is such that $\phi_{\t}\circ f\in B$ for all $\t\in \T$. Then there exist $a_1,a_2\in \F$ such that \begin{equation}\label{eq:a1}
        (a_1,0,a_1x)=(2,3,4)(\phi_{1}\circ f)=((\{1,2\})f,(\{1,3\})f,(\{1,4\})f)\end{equation}
    \begin{equation}\label{eq:a2}
        (a_2,0,a_2x)=(2,3,4)(\phi_{(1\,2)(4\,5)}\circ f)=((\{1,2\})f,(\{2,3\})f,(\{2,5\})f)\end{equation}
    It follows that $(\{1,3\})f=(\{2,3\})f=0$ and
    \begin{equation}\label{eq:exampleconstraints1}
        a_1=a_2=(\{1,2\})f=(\{1,4\})f\cdot x^{-1}=(\{2,5\})f\cdot x^{-1}
    \end{equation}
    Next, there exists $a_3\in\F$ such that
    \begin{equation}%\label{eq:a3}
        (a_3,0,a_3x)=(2,3,4)(\phi_{(1\,3)(4\,6)}\circ f)=((\{2,3\})f,(\{1,3\})f,(\{3,6\})f).
        \end{equation}
    But $(\{2,3\})f=0$ by~\eqref{eq:a2}, and so $a_3=0$ whence $(\{3,6\})f=0$ as well. Next, for some $a_4\in\F$, \begin{equation}\label{eq:a4}
        (a_4,0,a_4x)=(2,3,4)(\phi_{(1\,5)(2\,4)(3\,6)}\circ f)=((\{4,5\})f,(\{5,6\})f,(\{2,5\})f),\end{equation}
    and combining this with~\eqref{eq:exampleconstraints1} we see that $a_4=a_1.$ Finally, for some $a_5\in\F$,
    \begin{equation}\label{eq:a5}
        (a_5,0,a_5x)=(2,3,4)(\phi_{(1\,6)(2\,5)(3\,4)}\circ f)=((\{5,6\})f,(\{4,6\})f,(\{3,6\})f),\end{equation}
    but $(\{5,6\})f=0$ by \eqref{eq:a4}, and so $a_5=0$. This completes the constraints imposed by $\mathcal{T}_0$.

    Now, set $\mathcal{T}_1=\mathcal{T}_0\cup\{(1\,4)(2\,6)(3\,5)\}$ and $\mathcal{T}_2=\mathcal{T}_0\cup\{(1\,4)(2\,5)(3\,6)\}$. If $\phi_{(1\,4)(2\,6)(3\,5)}\circ f\in B$, then for some $a\in\F$, \begin{equation}
        (a,0,ax)=(2,3,4)(\phi_{(1\,4)(2\,6)(3\,5)}\circ f)=((\{4,6\})f,(\{4,5\})f,(\{1,4\})f).
    \end{equation} 
        But then $a=0$ by~\eqref{eq:a5} and so $(\{1,4\})f=0$ which implies that $a_1=0$, and so $\mathbf{G}_\Gamma(G,\mathcal{T}_1,1,B)=0.$

    On the other hand, if $\phi_{(1\,4)(2\,5)(3\,6)}\circ f\in B$, then for some $a'\in\F$, \begin{equation}
        (a',0,a'x)=(2,3,4)(\phi_{(1\,4)(2\,5)(3\,6)}\circ f)=((\{4,5\})f,(\{4,6\})f,(\{1,4\})f).\end{equation}
    This now yields that $a'=a_1$, and there are no further restrictions. Therefore, under lexicographic ordering we have that $\mathbf{G}_\Gamma(G,\mathcal{T}_2,1,B)=\larry (1,0,x,0,x,0,1,0,0)\rangle> \mathbf{G}_\Gamma(G,\mathcal{T}_1,1,B)=0.$
    % \begin{figure}[h]
    %     \centering
    %     \includegraphics[width=0.3\linewidth]{3-prism.png}
    %     \caption{The 3-prism.}
    %     \label{fig:3prism}
    % \end{figure}
\end{example}

%A code $C\leq \F^X$ is called \emph{symmetric} if there is some $H\leq\Aut(C)$ which is transitive on $X$. Such a code is \emph{single-orbit symmetric} if there is some $c\in C^\perp$ such that $C^\perp=\langle c^{H}\rangle$. A family of codes is \emph{highly symmetric} if it is single-orbit symmetric such that $c$ can be chosen of bounded weight---such a family of codes is necessarily LDPC.

Another easy consequence of Theorems~\ref{thm:AutomorphismsU} and~\ref{thm:AutomorphismsD} is that we can often verify that (di)graph codes are symmetric or even highly symmetric just from their ingredient (di)graph. The following corollary is one of the few times that the result for the digraph code is strictly weaker than the corresponding result for the graph code.

\begin{corollary}\label{cor:symmetry}
    Assume Hypotheses~\ref{hyp:GorD}. Then the  following hold: \begin{itemize}
        \item[(U)] If $G_\alpha^{\Gamma(\alpha)}$ is a transitive subgroup of $\Aut(B)$ then $C$ is symmetric. If, in addition, $B$ is single-orbit symmetric with respect to $G_\alpha^{\Gamma(\alpha)}$, then $C$ is highly symmetric with $C^\perp$ generated by $c^G$ for some codeword $c$ of weight at most $|\Gamma(\alpha)|$.
        \item[(D)] Suppose that $G_\alpha^{\Gamma^+(\alpha)}\leq\Aut(B_1)$ and $G_\alpha^{\Gamma^-(\alpha)}\leq \Aut(B_2)$. If $G_\alpha^{\Gamma^+(\alpha)}$ is transitive then $C$ is symmetric. If, in addition, $B_1$ and $B_2$ are single-orbit symmetric with respect to $G_\alpha^{\Gamma^+(\alpha)}$ and $G_\alpha^{\Gamma^-(\alpha)}$, respectively, then $C^\perp=\langle c_1^g,c_2^g : g \in G\rangle$ for some codewords $c_1$ and $c_2$ of weight at most $|\Gamma^+(\alpha)|=|\Gamma^-(\alpha)|.$
    \end{itemize}
\end{corollary}
\begin{proof}
    That the codes are symmetric follows immediately from part (3) of Theorems~\ref{thm:AutomorphismsU} and~\ref{thm:AutomorphismsD} using the fact that if $G$ is transitive on vertices and $G_\alpha$ is transitive on (out-)neighbours, then $G$ is transitive on edges (resp. arcs).

    Assume now that $\Gamma$ is undirected, and suppose that $B$ is single-orbit symmetric. Thus there exists $b\in B^\perp$ such that $B^\perp=\langle b^{G_\alpha}\rangle$. Define \[c :E\Gamma\to\F\text{ by }\{\alpha,\beta\}\mapsto (\beta)b \text{ and } e\mapsto 0\text{  whenever }\alpha\not\in e.\] 
    Now, since $B^\perp=\larry b^{G_\alpha}\roger$, we see from the definition of the dual code that for any $\t\in\T$, $\phi_\t\circ f\in B$ if and only if, for all $\g\in G_\alpha$,
    \begin{equation}\label{e:dual}
    \sum_{\beta\in\Gamma(\alpha)}\left((\beta)b^\g\cdot
        (\beta)(\phi_\t\circ f)\right)=0.
    \end{equation}
    Therefore, $f\in C$ if and only if \eqref{e:dual} holds 
% $$\sum_{\beta\in\Gamma(\alpha)}\left((\beta)b^\g\cdot(\beta)(\phi_\t\circ f)\right)=0$$ 
    for all $\g\in G_\alpha$ and $\t\in\mathcal{T}$. However, since $(\beta)b^\g=(\{\alpha,\beta\})c^\g$ for all $\beta\in \Gamma(\alpha)$, we deduce that
    \begin{align*}
        \sum_{\beta\in\Gamma(\alpha)}\left((\beta)b^\g\cdot(\beta)(\phi_\t\circ f)\right)=\sum_{\beta\in\Gamma(\alpha)}\left((\beta)b^\g\cdot(\beta)(\phi_1\circ f^\t)\right)&=\sum_{\beta\in\Gamma(\alpha)}\left((\{\alpha,\beta\})c^\g\cdot(\{\alpha,\beta\}) f^\t\right)\\&=\sum_{e\in E\Gamma}(e)c^\g\cdot(e) f^{\t}\\&=\sum_{e\in E\Gamma}(e)c^{\g\t^{-1}}\cdot(e)f.
    \end{align*} Hence $f\in C$ if and only if $f$ is orthogonal to each element of $\{c^{\g\t^{-1}} : \g\in G_\alpha,\t\in\T\}\subseteq c^G.$ Additionally, $C^\perp$ is invariant under $G^{E\Gamma}\leq\Aut(C)$, whence $\langle c^G\rangle= C^\perp$, proving (U).

    Finally, suppose $\Gamma$ is directed, and that $B_1$ and $B_2$ are single-orbit symmetric. Take $b_1\in B_1^\perp$ and $b_2\in B_2^\perp$ such that $B_1^\perp=\langle b_1^{G_\alpha}\rangle$ and $B_2^\perp=\langle b_2^{G_\alpha}\rangle$. Define $c_1 :A\Gamma\to\F$ by $$(\beta,\gamma)\mapsto\begin{cases}
        (\gamma)b_1&\text{ if $\beta=\alpha$}\\
        0&\text{ otherwise,}
    \end{cases}$$ and $c_2 :A\Gamma\to\F$ by $$(\beta,\gamma)\mapsto\begin{cases}
        (\beta)b_2&\text{ if $\gamma=\alpha$}\\
        0&\text{ otherwise.}
    \end{cases}$$
    Since $B_1$ and $B_2$ are single-orbit symmetric with respect to $G_\alpha$, we see that $f\in C$ if and only if $$\sum_{\beta\in\Gamma^+(\alpha)}(\beta)b_1^\g\cdot(\beta)(\varphi_\t\circ f)=0=\sum_{\beta\in\Gamma^-(\alpha)}(\beta)b_2^\g\cdot(\beta)(\psi_\t\circ f)$$ for all $\g\in G_\alpha$ and $\t\in\T$. By a  similar argument to the undirected case, the leftmost expression is equal to 0 exactly when $\sum_{a\in A\Gamma}(a)c_1^{\g\t^{-1}}\cdot(a)f=0$, and the rightmost is equal to 0 if and only if $\sum_{a\in A\Gamma}(a)c_2^{\g\t^{-1}}\cdot(a)f=0$, whence (D) holds.
\end{proof}
\begin{rk}
    The main application of each statement of Corollary~\ref{cor:symmetry} is that if we are careful with our choices of input, we can use our constructions to build families of symmetric LDPC codes, and in the graph code situation we can even get highly-symmetric families.
\end{rk}\medskip
%From now on, provided that $G$ acts transitively on vertices and that $G_\alpha^{\Gamma(\alpha)}\leq\Aut(B)$ we omit $\mathcal{T}$ and $\alpha$ when defining $\mathbf{C}_\Gamma$.
%An immediate consequence of Proposition~\ref{prop:DtoG} is the following proposition---it hints that if looking for a ``good" family of codes constructed from some family of undirected graphs, one is probably better off using the graph codes rather than the digraph codes.

%\begin{proposition}
   % Let $\Gamma=(V,E)$ be an undirected graph, let $\alpha\in V$, let $B\leq \F^{\Gamma(\alpha)}$, and let $\mathcal{T}$ be a transversal for $G_\alpha$ in $G$. Set $C=\mathbf{G}(\mathcal{T},\alpha,B)$ and %$D:=\mathbf{D}(\mathcal{T},\alpha,B)$. %The following hold:
    %\begin{itemize}
        %\item[(1)] $r(D)\geq (1/2)r(C)$; %and
        %\item[(2)] $\delta(D)\leq %\delta(C)$.
 %   \end{itemize}
%\end{proposition}
\section{Isomorphism, connectivity, and decompositions}\label{sec:isocondec}
In this section we investigate how the structure of the input (di)graph impacts the (di)graph code. We start by showing that the graph and digraph code constructions are invariant over isomorphic (di)graphs.

\begin{lemma}\label{lem:isoU}
    Assume Hypotheses~\ref{hyp:GorD} with $\Gamma$ undirected. Also let $\Gamma'=(V',E')$ be a graph isomorphic to $\Gamma$ and $\rho: \Gamma\to\Gamma'$ an isomorphism. Then
    \begin{enumerate}
        \item[(a)] The group $G'=\rho^{-1} G\rho$ is a vertex-transitive subgroup of $\Aut(\Gamma')$.
        \item[(b)] Setting $\alpha':=\alpha^\rho$ and $\T':=\T^\rho$, 
            \begin{enumerate}
                \item[(i)] the map $\tau_0:f\to \rho^{-1}|_{\Gamma(\alpha')}\circ f$ defines a vector space isomorphism $\F^{\Gamma(\alpha)}\to \F^{\Gamma(\alpha')}$ and $B':=(B)\tau_0$ is a linear code in $\F^{\Gamma(\alpha')}$;
                \item[(ii)] for $\t'=\rho^{-1}\t\rho\in\T'$ with $\t\in\T$, the image $(E_\t)^\rho$ is the set $E'_{\t'}$ of edges of $\Gamma'$ incident with $(\alpha')^{\t'}=(\alpha^\t)^\rho$; and $\phi^{\Gamma'}_{\t'}= \rho^{-1}|_{\Gamma(\alpha')}\circ \phi_\t^\Gamma\circ\rho$.
                \item[(iii)] the map $\tau:f\to \rho^{-1}\circ f$ defines a vector space isomorphism $\F^E\to \F^{E'}$ and the image $ C':=( C)\tau$ is equal to the graph code $\mathbf{G}_{\Gamma'}(G',\T', \alpha', B')$. 
            \end{enumerate}  
    \end{enumerate}
\end{lemma}
We again have an analogue for the digraph case; we omit the proof in the graph case.
\begin{lemma}\label{lem:isoD}
   Assume Hypotheses~\ref{hyp:GorD} with $\Gamma$ directed. Also let $\Gamma'=(V',A')$ be a digraph isomorphic to $\Gamma$ and $\rho: \Gamma\to\Gamma'$ an isomorphism. Then
    \begin{enumerate}
        \item[(a)] The group $G'=\rho^{-1} G\rho$ is a vertex-transitive subgroup of $\Aut(\Gamma')$.
        \item[(b)] Setting $\alpha':=\alpha^\rho$ and $\T':=\T^\rho$, 
            \begin{enumerate}
                \item[(i)] the maps $\tau^+:f\mapsto \rho^{-1}|_{\Gamma^+(\alpha')}\circ f$ and $\tau^-:f\mapsto \rho^{-1}|_{\Gamma^-(\alpha')}\circ f$ define vector space isomorphisms $\F^{\Gamma^+(\alpha)}\to \F^{\Gamma^+(\alpha')}$ and $\F^{\Gamma^-(\alpha)}\to \F^{\Gamma^-(\alpha')}$, respectively. Moreover, $B_1':=(B_1)\tau^+$ and $B_2':=(B_2)\tau^-$ are linear codes in $\F^{\Gamma^+(\alpha')}$ and $\F^{\Gamma^-(\alpha')}$, respectively;
                \item[(ii)] for $\t'=\rho^{-1}\t\rho\in\T'$ with $\t\in\T$, the image $(O_\t)^\rho$ is the set $O'_{\t'}$ of arcs with initial vertex $(\alpha')^{\t'}=(\alpha^\t)^\rho$ and the image $(I_\t)^\rho$ is the set $I'_{\t'}$ of arcs of $\Gamma'$  with terminal vertex $(\alpha')^{\t'}$; and $\varphi^{\Gamma'}_{\t'}= \rho^{-1}|_{\Gamma^+(\alpha')}\circ \varphi_\t^\Gamma\circ\rho$ and $\psi^{\Gamma'}_{\t'}= \rho^{-1}|_{\Gamma^-(\alpha')}\circ \psi_\t^\Gamma\circ\rho$.
                \item[(iii)] the map $\tau:f\mapsto \rho^{-1}\circ f$ defines a vector space isomorphism $\F^A\to \F^{A'}$ and the image $ C':=( C)\tau$ is equal to the digraph code $\mathbf{D}_{\Gamma'}(G',\T', \alpha', B_1',B_2')$. 
            \end{enumerate}  
    \end{enumerate}
\end{lemma}

\begin{proof}
    Part (a) follows as $\rho$ is a bijection $V\to V'$ which maps $A$ to $A'$, and $G$ is transitive on $V$.

    Part (b)(i) is straightforward, as is the first part of (b)(ii). For the last part of (b)(ii), let $\beta':=\beta^\rho\in \Gamma^{+}(\alpha')$ with $\beta\in \Gamma^+(\alpha)$. Then $\rho^{-1}|_{\Gamma^+(\alpha')}\circ \varphi_\t^\Gamma$ maps $\beta'$ to $(\alpha^t,\beta^t)$, and hence $\rho^{-1}|_{\Gamma^+(\alpha')}\circ \varphi_\t^\Gamma\circ\rho$ maps $\beta'$ to the pair consisting of $(\alpha^\t)^\rho=(\alpha')^{\rho^{-1} \t\rho}=(\alpha')^{\t'}$ and $(\beta^\t)^\rho=(\beta')^{\t'}$. Thus the claim for $\varphi_{\t'}^{\Gamma'}$ follows; the claim for $\psi_{\t'}^{\Gamma'}$ is identical.

    The first assertion of (b)(iii) is straightforward. Consider the image $C':=(C)\tau$, which is a subspace, and hence a linear code in $\F^{A'}$. An element $f'\in\F^{A'}$ is of the form $f'=(f)\tau = \rho^{-1}\circ f$ for some unique $f\in\F^A$, and $f'$ lies in $C'$ if and only if $f\in C$. The latter holds if and only if, for each $\t\in \T$, $\varphi_\t^\Gamma\circ f\in B_1$ and $\psi_\t^\Gamma\circ f\in B_2$. Using part (b)(ii) we have 
    \[
     \varphi^{\Gamma'}_{\t'}\circ f' = (\rho^{-1}|_{\Gamma^+(\alpha')}\circ \varphi_\t^\Gamma\circ\rho)\circ (\rho^{-1}\circ f) = \rho^{-1}|_{\Gamma^+(\alpha')}\circ \varphi^\Gamma_\t \circ f.
    \]
    The condition $\varphi^\Gamma_\t \circ f\in B_1$ is equivalent to $\rho^{-1}|_{\Gamma^+(\alpha')}\circ (\varphi^\Gamma_\t \circ f)$ lying in $(B_1)\tau=B_1'$ by part (b)(i). Thus $\varphi_\t^\Gamma\circ f\in B_1$ if and only if $\varphi^{\Gamma'}_{\t'}\circ f'\in B_1'$ and similarly for $B_2$ with respect to the maps $\psi_\t$. Part (b)(iii) now follows.%It follows that $f'\in C'$, that is to say $f\in  C$, if and only if $f'\in C_{\Gamma'}(\T',S',B')$, so $ C'= C_{\Gamma'}(\T',S',B')$. 
\end{proof}

\begin{remark}\label{r:noniso}
{\rm The converse does not hold. That is, even with a fixed input code $B$ of length $v$, one can find non-isomorphic vertex-transitive graphs of valency $v$ which yield the same graph code. This is easy to see for some `trivial' inputs, for example, the zero code, the full-space code, the repetition code, or any direct sum thereof. However, in fact, the same can be said even when restricting to `sensible' input choices. Indeed, we found computationally that there are non-isomorphic vertex-transitive cubic graphs on 20 points which yield the same graph code when the construction is fed the same input $B$. The input code we used was the binary augmentation code of length 3 (that is the code consisting of all weight-2 vectors), and the graphs were the 10-M\"{o}bius ladder and the unique graph of girth four, diameter four and order 20 (both of these graphs can be found in the database of cubic vertex-transitive graphs on up to 1280 vertices, see~\cite{CubeGraphs}). We  verified that these graphs were the smallest examples of non-isomorphic cubic graphs yielding the same code with a non-trivial input code over $\mathrm{GF}(2)$.
}
\end{remark}

Next we deal with vertex-transitive groups $G$ acting on disconnected (di)graphs $\Gamma$.

\begin{theorem}\label{thm:decompD}
    Let $\Gamma=(V,A)$ be a directed graph admitting a vertex-transitive subgroup $G\leq\Aut(\Gamma)$. Fix $\alpha\in V$ and let $\Sigma$ be the connected component of $\Gamma$ containing $\alpha$. Let $\mathcal{T}$ be a transversal for $G_\alpha$ in $G_{\{\Sigma\}}$ with $1\in\mathcal{T}$ and let $\mathcal{S}$ be a transversal for $G_{\{\Sigma\}}$ in $G$ with $1\in\mathcal{S}$. Then the product $\mathcal{TS}$ is a transversal for $G_\alpha$ in $G$, and for each $\s\in \mathcal{S}$, $\mathcal{T}^{\s}$ is a transversal for $G_{\{\Sigma^\s\}}$ in $G$. Moreover, for any $B_1\leq \F^{\Gamma^+(\alpha)}$ and $B_2\leq \F^{\Gamma^-(\alpha)}$ we have a direct sum decomposition \begin{equation}\label{eq:disconnect}
    \mathbf{D}_\Gamma(G,\mathcal{TS},\alpha,B_1,B_2)=\bigoplus_{\s\in \mathcal{S}} \mathbf{D}_{\Sigma^\s}(G_{\{\Sigma^\s\}},\mathcal{T}^\s,\alpha^\s,B_1^\s,B_2^\s),
    \end{equation}
    where $B_1^\s$ denotes the code obtained from the obvious isomorphism $\F^{\Gamma^+(\alpha)}\leftrightarrow\F^{\Gamma^+(\alpha)^\s}$ (and analogously $B_2^\s$). Additionally, all direct summands on the righthand side of~\eqref{eq:disconnect} are isomorphic.
\end{theorem}
\begin{proof}
    The transversal claims are clear.
    
    Observe that since $G\leq \Aut(\Gamma)$ is vertex-transitive, the components of $\Gamma$ are precisely the digraphs $\Sigma^\s$, and moreover, the out-neighbours (resp. in-neighbours) of $\alpha^\s$ are $\Gamma^+(\alpha)^\s$ (resp. $\Gamma^-(\alpha)^\s$). Let $f\in \F^{A\Gamma}$ and observe that for all $\t\in \mathcal{T}$, $\s\in \mathcal{S}$ and $\beta\in \Gamma^+(\alpha)$ we have $$(\beta)(\varphi_{\t\s}\circ f)=(\alpha^{\t\s},\beta^{\t\s})f=(\alpha^{\s\t^\s},\beta^{\s\t^\s})f=(\beta^\s)(\varphi_{\t^\s}\circ f).$$ An identical argument can be made for in-neighbours. In particular, $f\in \mathbf{D}_\Gamma(G,\mathcal{TS},\alpha,B_1,B_2)$ if and only if $f\mid_{A\Sigma^\s}\in\mathbf{D}_{\Sigma^\s}(G_{\{\Sigma^\s\}},\mathcal{T}^\s,\alpha^\s,B_1,B_2)$ for all $\s\in\mathcal{S}$. The claimed equation~\eqref{eq:disconnect} follows from the fact that the arcs sets $A\Sigma^\s$ form a partition of $A\Gamma$.

    Since $\Gamma$ is vertex-transitive, its components are pairwise isomorphic, with the isomorphism $\Sigma\leftrightarrow \Sigma^\s$ witnessed by the action of $\s$. Therefore the final claim follows from Lemma~\ref{lem:isoD}.
   % Let $s\in S$. To prove the final claim it suffices to show that $$C_1:=\mathbf{SD}_{\Sigma(\alpha^s)}(\mathcal{T}^s,\alpha^s,B_1,B_2)\cong \mathbf{SD}_{\Sigma(\alpha)}(\mathcal{T},\alpha,B_1,B_2)=:C_2.$$ Observe that for all $t\in\mathcal{T}$ and $\beta\in \Gamma^+(\alpha)$ we have that $$(\beta^s)(\varphi_{t^s}\circ f^{s^{-1}})=(\alpha^{ts},\beta^{ts})(f^{s^{-1}})=(\beta)(\varphi_{t}\circ f),$$ and a similar equality for the in-neighbours. Therefore the map $C_1\to C_2$ given by $f\mapsto f^{s^{-1}}$ defines an isomorphism of codes.
\end{proof}
Of course, with a similar proof we get the undirected analogue Theorem~\ref{thm:conndecompmain}. 
%\begin{theorem}\label{thm:decompU}
 %   Let $\Gamma=(V,E)$ be an undirected graph admitting a vertex-transitive subgroup $G\leq\Aut(\Gamma)$. Fix $\alpha\in V\Gamma$, and let $\mathcal{T}$ be a transversal for $G_\alpha$ in $G_{\{\Sigma(\alpha)\}}$ with $1\in\mathcal{T}$ and let $\mathcal{S}$ be a transversal for $G_{\{\Sigma(\alpha)\}}$ in $G$ with $1\in\mathcal{S}$. Then for any $B\leq \F^{\Gamma(\alpha)}$ we have a decomposition \begin{equation}\label{eq:disconnectU}\mathbf{G}_\Gamma(G,\mathcal{TS},\alpha,B)=\bigoplus_{\s\in \mathcal{S}} \mathbf{G}_{\Sigma(\alpha^\s)}(G_{\{\Sigma(\alpha)\}}^s,\mathcal{T}^\s,\alpha^\s,B^\s),\end{equation}
  %  where  all direct factors in are isomorphic.
%\end{theorem}
Theorems~\ref{thm:conndecompmain} and~\ref{thm:decompD} indicate that---for these constructions---we can restrict our focus to connected graphs; this generalises~\cite[Theorem 1.2]{AruPraRad}.

Using connected components is not the only natural way to decompose the edge- and arc-sets. We can also use an \emph{orbital decomposition}. Suppose we have a (un)directed graph $\Gamma$ with $G\leq\Aut(\Gamma)$ acting transitively on $V\Gamma$, and fix $\alpha\in V\Gamma$. Then $G$ induces an action on the arcs (resp. edges) of $\Gamma$, and hence induces a partition of the arcs (resp. edges) into $G$-orbits. We start by listing a couple of easy and standard facts about this decomposition.

\begin{lemma}\label{lem:orbitalfacts}
    Let $\Gamma$ be a (un)directed graph $\Gamma$ with $G\leq\Aut(\Gamma)$ acting transitively on $V\Gamma$, and fix $\alpha\in V\Gamma$. The following hold.
    \begin{itemize}
        \item[(1)] If $\Gamma$ is directed then each $G$-orbit on $A\Gamma$ contains an element of the form $(\alpha,\beta)$ where $\beta\in \Gamma^+(\alpha)$, and each such element is contained in exactly one $G$-orbit. If $\Gamma$ is undirected then each $G$-orbit  on $E\Gamma$ contains an element of the form $\{\alpha,\beta\}$ where $\beta\in \Gamma(\alpha)$, and each such element is contained in exactly one $G$-orbit.
        
        \item[(2)] Let us list the orbits of $G$ on arcs (resp. edges) as $O_1,O_2,\dots, O_m$. For $1\leq i\leq m$ define the \emph{orbital subgraph} $\Gamma_i$ to have vertex set $V\Gamma$ and arc set (resp. edge set) $O_i$. Then the action of $G$ on $\Gamma_i$ is vertex-transitive and arc-transitive (resp. edge-transitive).
    \end{itemize}
\end{lemma}
\begin{theorem}\label{thm:orbitaldecompD}
    Assume Hypotheses~\ref{hyp:GorD} with $\Gamma$ directed. Consider the decomposition of $A\Gamma$ given in Lemma~\ref{lem:orbitalfacts} and, for each $i\leq m$, write $\Gamma_i^+(\alpha)$ and $\Gamma_i^-(\alpha)$ for the out-neighbours and in-neighbours, respectively, of $\alpha$ in $\Gamma_i$. Then $\F^{\Gamma^+(\alpha)}$ decomposes as $\bigoplus_{i\leq m}\F^{\Gamma_i^+(\alpha)}$ and a similar decomposition exists for $\F^{\Gamma^-(\alpha)}$. Suppose that $B_1$ and $B_2$ decompose as direct sums $\bigoplus_{i\leq m}B_{1i}$ and $\bigoplus_{i\leq m}B_{2i}$, respectively, with $B_{1i}\leq \F^{\Gamma_i^+(\alpha)}$ and $B_{2i}\leq \F^{\Gamma_i^-(\alpha)}$ for each $i\leq m$. Then 
    $$
\mathbf{D}_\Gamma(G,\T,\alpha,B_1,B_2)=\bigoplus_{1\leq i\leq m}\mathbf{D}_\Gamma(G,\T,\alpha,B_{1i},B_{2i}).
$$ 
Moreover, if $G_\alpha^{\Gamma_i^+(\alpha)}\leq\Aut(B_{1i})$ and $G_\alpha^{\Gamma_i^-(\alpha)}\leq\Aut(B_{2i})$ for $1\leq i\leq m$, then each direct summand is a symmetric linear code.
\end{theorem}
\begin{proof}
    Let $f\in \F^{A\Gamma}$ and let $\t\in\T$. Then $(\beta)\varphi^\Gamma_\t\circ f=(\beta)\varphi^{\Gamma_i}_\t\circ f|_{A\Gamma_i}$, where $\Gamma_i$ is the (unique) orbital subgraph with $\beta\in\Gamma_i^+(\alpha)$. In particular, since the $A\Gamma_i$ partition $A\Gamma$, we deduce that $\varphi_\t^{\Gamma}\circ f\in B_1$ if and only if $\varphi^{\Gamma_i}_\t\circ f|_{A\Gamma_i}\in B_{1i}$ for all $i\leq m$. Similarly, $\psi^\Gamma_\t\circ f\in B_2$ if and only if $\psi^{\Gamma_i}_\t\circ f|_{A\Gamma_i}\in B_{2i}$ for all $i\leq m$. The claimed decomposition now follows from the definition of the digraph code.

    The second claim follows from Theorem~\ref{thm:AutomorphismsD} and Lemma~\ref{lem:orbitalfacts}(2).
\end{proof}

As usual, an undirected version of Theorem~\ref{thm:orbitaldecompD} follows with an analogous proof.
\begin{theorem}\label{thm:orbitaldecompU}
    Assume Hypotheses~\ref{hyp:GorD} with $\Gamma$ undirected. Consider the decomposition of $E\Gamma$ given in Lemma~\ref{lem:orbitalfacts} and write $\Gamma_i(\alpha)$ for the neighbours of $\alpha$ in $\Gamma_i$ so that $\F^{\Gamma(\alpha)}$ decomposes as $\bigoplus_{i\leq m}\F^{\Gamma_i(\alpha)}$. Suppose that $B$ decomposes as a direct sum $\bigoplus_{i\leq m}B_{i}$ with $B_{i}\leq \F^{\Gamma_i(\alpha)}$ for each $i\leq m$. Then $$\mathbf{G}_\Gamma(G,\T,\alpha,B)=\bigoplus_{1\leq i\leq m}\mathbf{G}_\Gamma(G,\T,\alpha,B_{i}).$$ Moreover, if $G_\alpha^{\Gamma_i(\alpha)}\leq\Aut(B_{i})$ for $1\leq i\leq m$, then each direct summand is a symmetric linear code.
\end{theorem}

\section{Statistics of digraph codes}\label{sec:stats}
We begin this section by proving a generalisation of the Alon--Chung Expander Mixing Lemma from~\cite{AlonChung}, adapted for digraphs. We remark that this result is probably well-known, but we have failed to locate it in the literature. Given a real symmetric matrix $M$ we may denote its eigenvalues as $\lambda_1,\lambda_2,\dots,\lambda_n$ such that 
$$
|\lambda_1|\geq|\lambda_2|\geq\cdots\geq |\lambda_n|.
$$ 
Labelled in this way, we write $\lambda(M):=|\lambda_2|$, the \emph{absolute value of the eigenvalue of $M$ of second largest magnitude}.

\begin{proposition}\label{prop:mixinglem}
    Let $\Gamma=(V,A)$ be a connected directed graph with $|\Gamma^+(\alpha)|=|\Gamma^-(\alpha)|=v$ for all $\alpha\in V$. Let $M$ be the adjacency matrix of $\Gamma$, set $\hat{M}=M+M^{T}$ and let $\lambda=\lambda(\hat{M})$. For any $\epsilon$ such that $0<\epsilon<1$ and any $X\subseteq V$ of size $\epsilon n$, the subgraph of $\Gamma$ induced on $X$ has at most $vn\cdot(\epsilon^2 +\frac{\lambda}{2v}\left(\epsilon-\epsilon^2)\right)$ arcs.
\end{proposition}
\begin{proof}
    Let $\mathbbm{1}_V$ denote the all-$1$ row  vector of $\mathbb{R}^V$, and note that $\mathbbm{1}_V\hat{M}=\mathbbm{1}_VM+\mathbbm{1}_VM^\top=2v\mathbbm{1}_V$, and that $2v$ is the largest eigenvalue of $\hat{M}$. Moreover, since $\Gamma$ is connected the $2v$-eigenspace has dimension one; that is, the $2v$-eigenspace is $\larry \mathbbm{1}_V\roger$. Additionally, since $\hat{M}$ is a real symmetric matrix it has a basis of eigenvectors, and every two of its distinct eigenspaces are orthogonal~\cite[Lemma 8.4.1]{AlgGraTh}. Denote by $a(X)$ the number of arcs in the induced sub-digraph on $X$, and for each subset  $S\subseteq V$ let $\mathbbm{1}_S$ denote the characteristic (row) vector for $S$. Then $a(X)=\mathbbm{1}_X M\mathbbm{1}_X^\top=\mathbbm{1}_X M^\top\mathbbm{1}_X^\top,$ whence \begin{equation}\label{e:arccount}
        2a(X)=\mathbbm{1}_X\hat{M}\mathbbm{1}_X^\top.
    \end{equation} 
    For the subset $X$ in the statement, define $u\in \mathbb{R}^V$ such that $u_i=\begin{cases}
        1-\epsilon&\text{ if $i\in X$}\\
        -\epsilon&\text{ if $i\not\in X$}.
    \end{cases}$.
    Then \begin{equation}\label{e:orthog}
        \sum_{i\in V}u_i=\sum_{i\in X}u_i+\sum_{i\in V\setminus X}u_i=\epsilon n(1-\epsilon)-(n-\epsilon n)\epsilon=0,
    \end{equation} and so $u$ is orthogonal to the $2v$-eigenspace of $\hat{M}$. Thus, $u$ can be expressed as a linear combination of eigenvectors of $\hat{M}$, each of which is orthogonal to $\mathbbm{1}_V$. Consequently, \begin{equation}\label{e:2ndeval}
        |u \hat{M}u^\top|\leq \lambda |uu^\top|=\lambda |\epsilon n(1-\epsilon)^2+(n-\epsilon n)\epsilon^2|=\lambda\epsilon n(1-\epsilon).
    \end{equation}
    Next, observe that $\mathbbm{1}_X=\epsilon \mathbbm{1}_V+u$, and so putting this together with~\eqref{e:arccount}, and the fact that $\mathbbm{1}_V \hat{M}\mathbbm{1}_V^\top= 2v \mathbbm{1}_V \mathbbm{1}_V^\top=2vn$, we get \begin{align}2a(X)=(\epsilon \mathbbm{1}_V+u)\hat{M}(\epsilon \mathbbm{1}_V+u)^\top&=\epsilon^2 \mathbbm{1}_V\hat{M}\mathbbm{1}_V^\top+\epsilon u\hat{M}\mathbbm{1}_V^\top+\epsilon \mathbbm{1}_V\hat{M}u^\top +u\hat{M}u^\top\nonumber\\&\leq \epsilon^2\cdot 2vn+\epsilon u(\hat{M}\mathbbm{1}_V^\top)+\epsilon (\mathbbm{1}_V\hat{M})u^\top + \lambda\epsilon n(1-\epsilon)&\text{ by~\eqref{e:2ndeval},}\nonumber\\
    &=\epsilon^2\cdot 2vn+\epsilon\cdot 2v (u\mathbbm{1}_V^\top)+\epsilon\cdot 2v (\mathbbm{1}_Vu^\top) + \lambda\epsilon n(1-\epsilon)\nonumber\\&=\epsilon^2\cdot 2vn+ \lambda\epsilon n(1-\epsilon)& \text{by~\eqref{e:orthog}}.\nonumber\end{align}
    The result now follows.
\end{proof}

We are now ready to prove a bound on the relative distance $\delta(C)$ of a digraph code $C$.

\begin{theorem}\label{thm:distanceboundD}
    Assume Hypotheses~\ref{hyp:GorD} with $\Gamma$ directed. Set $v=|\Gamma^{+}(\alpha)|=|\Gamma^{-}(\alpha)|$, let $M$ be the adjacency matrix of $\Gamma$. Suppose that $i\in\{1,2\}$ is such that $\delta(B_i) = \min\{ \delta(B_1), \delta(B_2)\}$, and that $\lambda=\lambda(M+M^\top)<2v$. Then $$\delta(C)\geq \frac{\delta(B_i)}{2}\cdot \left(\frac{\delta(B_i)-(\lambda/v)}{2-(\lambda/v)}\right).$$
\end{theorem}

\begin{proof}
    Note that $\lambda<2v$ so the right hand side has positive denominator. Assume that $\delta(B_i)> \lambda/v$ as otherwise there is nothing to prove. Suppose, for a contradiction, that there exists non-zero $c\in C$ of weight $$\mathrm{wt}(c)<vn\frac{\delta(B_i)}{2}\cdot \left(\frac{\delta(B_i)-(\lambda/v)}{2-(\lambda/v)}\right).$$ Set $\epsilon_0=\frac{\delta(B_i)-(\lambda/v)}{2-(\lambda/v)}> 0$ so that $\mathrm{wt}(c)<vn\frac{\delta(B_i)}{2}\epsilon_0$ and \begin{align}\begin{split}\label{eq:wtcalc}\left(\epsilon_0 +\frac{\lambda}{2v}(1-\epsilon_0)\right)&=\left(\frac{\delta(B_i)-(\lambda/v)}{2-(\lambda/v)}+\frac{\lambda}{2v}\left(\frac{2-\delta(B_i)}{2-(\lambda/v)}\right)\right)\\&=\left(\frac{\delta(B_i)-(\lambda/2v)\delta(B_i)}{2-(\lambda/v)}\right)\\
    &=\left(\frac{\delta(B_i)-(\lambda/v)}{2-(\lambda/v)}\right)\\ 
    &= \frac{\delta(B_i)}{2} 
    .\end{split}\end{align}
    Thus, $vn\cdot\left(\epsilon_0^2+\frac{\lambda}{2v}(\epsilon_0-\epsilon_0^2)\right)>\mathrm{wt}(c)$. Since $f(\epsilon):=\epsilon^2+\frac{\lambda}{2v}(\epsilon-\epsilon^2)$ is increasing on $(0,\epsilon_0)$ it follows that there exists $\epsilon$ such that $0<\epsilon<\epsilon_0$ and $\mathrm{wt}(c)=vn\cdot(\epsilon^2 +\frac{\lambda}{2v}\left(\epsilon-\epsilon^2)\right)$. Consider the set of non-zero entries of $c$. Each such entry corresponds to an arc of $\Gamma$, and these arcs cover a set $V_0\subseteq V$ of at least $\epsilon n$ vertices by Proposition~\ref{prop:mixinglem}.

    The number of incident arc-vertex pairs in the induced subgraph on $V_0$ is $2\mathrm{wt}(c)$ (counting first the arc), and by counting first the vertex, this is equal to 
    $$
    \sum_{\beta\in V_0}\left(|\Gamma^+(\beta)\cap V_0|+|\Gamma^-(\beta)\cap V_0|\right).
    $$%the sum over the vertices of $V_0$ of their number of (in- and out-) neighbours in $V_0$. 
    Thus the average of $|\Gamma^+(\beta)\cap V_0|+|\Gamma^-(\beta)\cap V_0|$ over $\beta\in V_0$ is $2\mathrm{wt}(c)/|V_0|\leq 2\mathrm{wt}(c)/\epsilon n$, and so there exists some vertex $\beta\in V_0$ such that $|\Gamma^+(\beta)\cap V_0|+|\Gamma^-(\beta)\cap V_0|\leq 2\mathrm{wt}(c)/\epsilon n$. Thus we have 
    \begin{align*}
    |\Gamma^+(\beta)\cap V_0|+|\Gamma^-(\beta)\cap V_0|\leq 2\mathrm{wt}(c)/\epsilon n&=2vn\cdot\left(\epsilon^2 +\frac{\lambda}{2v}(\epsilon-\epsilon^2)\right)/\epsilon n\\&=2v\cdot\left(\epsilon +\frac{\lambda}{2v}(1-\epsilon)\right)\\&<2v\cdot\left(\epsilon_0+\frac{\lambda}{2v}(1-\epsilon_0)\right)\\&= v\delta(B_i),
    \end{align*} 
    where the last equality uses~\eqref{eq:wtcalc}.
    % arcs corresponding to non-zero entries of the codeword $c$. 
    Writing $\beta=\alpha^\t$ for some $\t\in\T$, this means that both $\mathrm{wt}(\varphi_\t\circ c)<\delta(B_i)v$ and $\mathrm{wt}(\psi_\t\circ c)<\delta(B_i)v$ whence $\varphi_\t\circ c=\psi_\t\circ c=0$ by the definition of $\delta(B_i)$. This contradicts the definition of $V_0$, and the result follows.
\end{proof}

For various types of codes constructed from graphs (such as Cayley codes) there is an analogous lower-bound in the literature of the shape $\left((\delta(B)-\lambda/v)/(1-\lambda/v)\right)^2$. This type of bound was first given in~\cite{SipSpiel}, and has become the standard bound in the Cayley code set-up. With an identical proof to that for Theorem~\ref{thm:distanceboundD}, but using the standard Alon--Chung Mixing Lemma~\cite[Lemma 2.3]{AlonChung} we are able to slightly strengthen this bound for graph codes, and hence for Cayley codes.

\begin{theorem}\label{thm:distanceboundU}
    Assume Hypotheses~\ref{hyp:GorD} with $\Gamma$ undirected. Set $v=|\Gamma(\alpha)|$, let $M$ be the adjacency matrix of $\Gamma$, and suppose that $\lambda:=\lambda(M)< v$. Then $$\delta(C)\geq \delta(B)\cdot \left(\frac{\delta(B)-(\lambda/v)}{1-(\lambda/v)}\right).$$
\end{theorem}
Despite being so similar to the proof of Theorem~\ref{thm:distanceboundD}, we have decided to make explicit the proof for Theorem~\ref{thm:distanceboundU}, since it provides an improvement on a frequently used bound.
\begin{proof}[Proof of Theorem~\ref{thm:distanceboundU}]
     Assume that $\delta(B)\geq \lambda/v$, as otherwise there is nothing to prove. Suppose, for a contradiction, that there exists non-zero $c\in C$ of weight $$\mathrm{wt}(c)<\frac{vn}{2}\delta(B)\cdot \left(\frac{\delta(B)-(\lambda/v)}{1-(\lambda/v)}\right).$$ Set $\epsilon_0=\frac{\delta(B)-(\lambda/v)}{1-(\lambda/v)}$ and observe that \begin{align*}\frac{vn}{2}\cdot\left(\epsilon_0^2 +\frac{\lambda}{v}(\epsilon_0-\epsilon_0^2)\right)&=\frac{vn}{2}\cdot\left( \frac{\delta(B)-(\lambda/v)}{1-(\lambda/v)}\right)\cdot\left(\frac{\delta(B)-(\lambda/v)}{1-(\lambda/v)}+\frac{\lambda}{v}\left(\frac{1-\delta(B)}{1-(\lambda/v)}\right)\right)\\&=\frac{vn}{2}\cdot \left(\frac{\delta(B)-(\lambda/v)}{1-(\lambda/v)}\right)\cdot\left(\frac{\delta(B)-(\lambda/v)\delta(B)}{1-(\lambda/v)}\right)\\&=\frac{vn}{2}\delta(B)\cdot \left(\frac{\delta(B)-(\lambda/v)}{1-(\lambda/v)}\right)\\&>\mathrm{wt}(c).\end{align*}
    It follows that there exists some $\epsilon$ such that $0<\epsilon<\epsilon_0$ and $\mathrm{wt}(c)=\frac{vn}{2}\cdot(\epsilon^2 +\frac{\lambda}{v}\left(\epsilon-\epsilon^2)\right)$. Each non-zero entry of $c$ corresponds to an edge of $\Gamma$, and these edges cover a set $V_0$ of at least $\epsilon n$ vertices by the Alon--Chung Mixing Lemma~\cite[Lemma 2.3]{AlonChung}.

    By the hand-shaking lemma, the sum over the vertices of $V_0$ of their valencies in the induced subgraph on $V_0$ is $2\mathrm{wt}(c)$. In particular, there is some vertex $\beta=\alpha^\t\in V_0$ of valency (in $V_0$) at most $2\mathrm{wt}(c)/|V_0|\leq 2\mathrm{wt}(c)/\epsilon n$, which is $$vn\cdot\left(\epsilon^2 +\frac{\lambda}{v}(\epsilon-\epsilon^2)\right)/\epsilon n=v\cdot\left(\epsilon +\frac{\lambda}{v}(1-\epsilon)\right)<v\cdot\left(\epsilon_0+\frac{\lambda}{v}(1-\epsilon_0)\right)=v\delta(B).$$ But then $\mathrm{wt}(\phi_\t\circ c)<\delta(B)v$ whence $\phi_\t\circ c=0$ by the definition of $\delta(B)$, a contradiction.
\end{proof}

Using this bound we provide a modest improvement to~\cite[Theorem 11]{KaufLub}. The following result is obtained from the same construction as Kaufman and Lubotzky (we restrict to the construction for the projective \emph{special} linear groups, so as to avoid bipartite graphs), but using Theorem~\ref{thm:distanceboundU} rather than the relative distance bound provided in~\cite[Theorem 2]{KaufLub}.
\begin{corollary}\label{cor:KauLuBound}
    Fix some $a\geq 8$ and $q=4093$. There exists an asymptotically good family of highly symmetric LDPC codes of rate at least $2/a$ and relative distance at least $$\frac{a-2}{2a\ln(q+1)}\cdot\left(\frac{\frac{a-2}{2a\ln(q+1)}-\frac{2\sqrt{q}}{q+1}}{1-\frac{2\sqrt{q}}{q+1}}\right)$$
\end{corollary}
\begin{remark}
    {\rm Depending on the choice of $a$, this bound on the relative distance is between roughly two to three times the lower bound which was used by Kaufman and Lubotzky~\cite{KaufLub}.}
\end{remark}

We now move on to investigating the rates of graph codes and digraph codes; this is a much easier endeavour.
\begin{theorem}\label{thm:ratebound}
    Assume Hypotheses~\ref{hyp:GorD}. If $\Gamma$ is undirected then $r(C)\geq 2r(B)-1$ and if $\Gamma$ is directed then $r(C)\geq r(B_1)+r(B_2)-1.$
\end{theorem}
\begin{proof}
    Assume that $\Gamma$ is directed; the undirected case is similar. Suppose $\Gamma$ has out-valency $v$. Every codeword in $C$ obeys exactly $|\mathcal{T}|$ `constraints' corresponding to the $\varphi_\t$, and another $|\mathcal{T}|$ corresponding to the $\psi_\t$. The constraint corresponding to $\varphi_\t$ (resp. $\psi_\t$) imposes $v-vr(B_1)$ (resp. $v-vr(B_2)$) linear constraints of $B_1$ (resp. $B_2$). Thus, the total number of linear constraints is at most $v\cdot|\mathcal{T}|\cdot(2-(r(B_1)+r(B_2)))$, that is, $C$ has dimension at least $v|\mathcal{T}|-v|\mathcal{T}|(2-(r(B_1)+r(B_2)))$. Then, since codewords of $C$ have length $v|\mathcal{T}|$, we have $$r(C)\geq r(B_1)+r(B_2)-1,$$ as asserted.
\end{proof}

\section{Examples and constructions}\label{sec:examples}
In this section we examine in more detail some properties of our constructions when restricting to certain families of graphs and codes.

\subsection{Complete bipartite graph codes}\label{sec:bipartite}
It turns out that the graph code construction generalises other well-known constructions of codes. Recall the definition of the direct product of codes from Definition~\ref{def:prodandsum}

\begin{theorem}\label{thm:tensorproduct}
    Let $m\geq2$ be an integer and set $\Gamma=K_{m,m}^\to$ the complete bipartite graph on $2m$ vertices viewed as a directed graph with in- and out-valency $m$, and $V\Gamma=V_0\sqcup V_1:=\{1,2,\dots,m\}\sqcup\{m+1,m+2,\dots,2m\}$ with the bipartition indicated by the union. Let $G=\Aut(\Gamma)=S_m\wr S_2$, and let $B_1,B_2\leq \F^{V_1}$. Finally, let $$\T=\{1,(1\,2),(1\,3),\dots, (1\,m),\tau,(1\,2)\tau,(1\,3)\tau,\dots,(1\,m)\tau\},$$ where $\tau$ is the unique element of $G$ which swaps $\alpha$ and $m+\alpha$ for all $\alpha\in V_0$. Then $C:=\mathbf{D}_\Gamma(G,\mathcal{T},1,B_1,B_2)\cong(B_2\otimes B_1)\oplus (B_2\otimes B_1)$.
\end{theorem}
\begin{proof}
    Write $A^+=V_0\times V_1$ and $A^-=V_1\times V_0$ so that $A\Gamma=A^+\sqcup A^-$ (disjoint union). Let $c\in C$. For each $\alpha\in V_0$ there exists some $b\in B_1$ such that $\varphi_{(1\,\alpha)}\circ c=b\in B_1$, and so using the `$f_\alpha$' notation as in Definition~\ref{def:prodandsum} we obtain \begin{equation}\label{eq:productres1}
        (\beta)c_\alpha=(\alpha,\beta)c=(\beta)(\varphi_{(1\,\alpha)}\circ c)=(\beta)b \text{ for all }\beta\in V_1,
    \end{equation}
    that is, $(c|_{A^+})_\alpha\in B_1.$
    Let $\beta\in V_1$, then there exists $b'\in B_2$ such that
    
  %  Similarly, $\psi_{(1\,\beta)}\circ c\in B_2$ for each so for each $\beta\in V_0$ we have that there exists $b'\in B_2$ such that 
  \begin{equation}\label{eq:productres2}
        (\alpha)c_{\beta}=(\alpha,\beta)c=(\alpha+m)(\psi_{(1\,(\beta-m))\tau}\circ c)=(\alpha+m)b' \text{ for all } \alpha\in V_0, 
    \end{equation}
    Thus, $(c\mid_{A^+})_\beta\in B_2^\tau\cong B_2$. The restrictions in~\eqref{eq:productres1} and~\eqref{eq:productres2} are the only ones whose constraints affect the values of $c$ on $A^+$, whence $C|_{A^+}\cong B_2\otimes B_1.$ Similar arguments give that $C|_{A^-}\cong B_2\otimes B_1$, exhausting all constraints, whence $C=C|_{A^+}\oplus C|_{A^-}\cong (B_2\otimes B_1)\oplus(B_2\otimes B_1),$ as claimed. % map $\tau$ thus induces an isomorphism between $B_2$ and $\langle (c\mid_{A^+})_\beta :c\in C,\beta\in V_1\rangle$This exhausts all restrictions on values $(i,j)c$ with $1\leq i\leq m$ and $m+1\leq j\leq 2m$. The constraints imposed by the remaining $\varphi_\t$ and $\psi_\t$ apply only to the values $(i,j)c$ with $m+1\leq i\leq 2m$ and $1\leq j\leq m$, thus $C$ decomposes as a direct sum with first direct summand satisfying that each row is a codeword of $B_1$ by~\eqref{eq:productres1} and each column is a codeword of $B_2$ by~\eqref{eq:productres2}. That is, the first direct summand is $B_1\otimes B_2$. Working out the remaining constraints is similar, and one can easily verify that they imply that the second direct summand is $B_1\otimes B_2$, hence the result.
\end{proof}
We can also use the undirected graph code to witness the direct product of a code with itself; we omit the proof as it is similar to the previous proof.
\begin{theorem}\label{thm:tensorsquare}
    Let $m\geq2$ be an integer and set $\Gamma=K_{m,m}$ the undirected complete bipartite graph with vertices labelled as in Theorem~\ref{thm:tensorproduct}. Let $G=\Aut(\Gamma)=S_m\wr C_2$, and let $B\leq \F^{\Gamma(1)}$ be any linear code. Then $C:=\mathbf{G}_\Gamma(G,\T,1,B)=B\otimes B$ where $\T$ is as defined in Theorem~\ref{thm:tensorproduct}.
\end{theorem}
\begin{remark}{\rm
    Given that we can obtain a direct product of a code with itself as a graph code, one might ask whether the direct sum of a code with itself may also be constructed as a graph code. It turns out that this is not possible in general: if the input code has length $n$ then a hypothetical graph used to yield a direct sum would need to have valency $n$ and $2n$ edges. For such a graph with vertex set $V$, the number of arcs would be $n|V|=4n$, so $|V|=4$ yielding $n\leq 3$.}
\end{remark}

\subsection{Petersen graph codes}\label{ex:petersen}
The smallest vertex-transitive graph which is not a Cayley graph is the Petersen graph. Thus, using the Petersen graph we can work out the shortest length highly symmetric graph code (over $\F:=\mathrm{GF}(2)$) which can not be witnessed as a Cayley code. 

Let $\Gamma$ be the Petersen graph. Then $\Aut(\Gamma)$ is the symmetric group $S_5$ and its action on vertices is its natural action on (unordered) pairs of points from $\{1,\dots,5\}$. Let $\alpha\in V\Gamma$. Then $G_\alpha^{\Gamma(\alpha)}$ is the symmetric group $S_3$ on three vertices. Thus, keeping Theorem~\ref{thm:AutomorphismsU} in mind, in order to construct a symmetric code we shall take as input a code $B\leq\F^3$ which is invariant under all permutations of three points. Up to equivalence, this leaves four input choices: the zero code, the unique $[3,1,3]_2$-code (the \emph{repetition code} of length three), the full space code, and the unique $[3,2,2]_2$-code (the \emph{augmentation code} of length three). Since there exist Cayley graphs on ten vertices of valency three (for example, the 5-prism), it is immediately clear that the zero code and the full-space code would yield codes which can also be witnessed as Cayley codes. Similarly, it is not hard to check that, if a repetition code is used as the input code, then the graph code construction for any connected graph will always output a repetition code. Therefore, to construct a symmetric graph code which is not a Cayley code, for a graph of valency 3 (such as the Petersen graph), we need to take $B$ to be the augmentation code of length 3.

Since $G_\alpha^{\Gamma(\alpha)}=\Aut(B)$, the construction is independent of our choice of transversal by Corollary~\ref{cor:Tindep}. Thus we may fix any transversal $\T$ of $G_\alpha$ in $G$ with $1\in\T$ and define $C=\mathbf{G}(G,\T,\alpha,B)$. Now,  $B^\perp$ has dimension one and thus it is single-orbit symmetric (see Section~\ref{sec:codedefs}). Therefore, $C=\mathbf{G}(G,\T,\alpha,B)$ is highly symmetric by Corollary~\ref{cor:symmetry}. By computer we work out that the constructed code $C$ has generator matrix $$
\begin{bmatrix}
1 & 0 & 1 & 0 & 0 & 0 & 0 & 1 & 0 &0 & 0 & 0 & 1 & 1 & 1 \\
0 & 1 & 1 & 0 & 0 & 0 & 0 & 0 & 0 & 0 & 1 & 0 & 1 & 0 & 1 \\
0 & 0 & 0 & 1 & 0 & 1 & 0 & 0 & 0 & 0 & 0 & 0 & 1 & 1 & 1 \\
0 & 0 & 0 & 0 & 1 & 1 & 0 & 0 & 0 & 0 & 1 & 1 & 0 & 1 & 0 \\
0 & 0 & 0 & 0 & 0 & 0 & 1 & 1 & 0 & 1 & 0 & 0 & 0 & 1 & 1 \\
0 & 0 & 0 & 0 & 0 & 0 & 0 & 0 & 1 & 1 & 1 & 1 & 0 & 0 & 1
\end{bmatrix}
.$$
Thus, by our discussion above, this $[15,6,5]_2$-code is a shortest highly symmetric graph code which is not a Cayley code (verified computationally) and up to equivalence is the unique such code satisfying $\overline{B}=B$ for the input code $B$ (with $\overline{B}$ as in Proposition~\ref{prop:Bbar}).

\subsection{Good digraph codes}\label{sec:goodcodes}
We wrap up this section with a result to demonstrate the viability of the digraph code construction as a tool to build good families of codes. Call a digraph $\Gamma$ a \emph{proper digraph} if there is some $(\alpha,\beta)\in A\Gamma$ with $(\beta,\alpha)\not\in A\Gamma$. We call a digraph code a \emph{proper digraph code} if it is constructed from a proper digraph.

The following theorem was inspired by the construction in \cite{LuPhiSar}. It will be proved over a series of lemmas. 
\begin{theorem}\label{thm:properdig}
    Set $p=4093$ (which is a prime) and let $q\equiv1\pmod{4}$ be prime such that $p\ne q$ and $p$ is a quadratic residue modulo $q$. Then there exists a proper digraph code $C$ of length $(p+1)q(q^2-1)/4$ with rate $r(C)\geq 2/(p+1)\approx0.0005$ and relative distance $$\delta(C)\geq \frac{213}{p+1}\cdot\left(\frac{213-2\sqrt{p}}{p+1-2\sqrt{p}}\right)\approx0.001.$$
    Therefore, there is an infinite family of good proper digraph codes.
\end{theorem}
Our first lemma simply guarantees the existence of infinitely many primes $q$ as in the statement of Theorem~\ref{thm:properdig}; it follows immediately from the Chinese Remainder Theorem and Dirichlet's Theorem on primes in arithmetic progressions.
\begin{lemma}\label{lem:properdig}
    Let $p=4093$. There are infinitely many primes $q\equiv1\pmod{4}$ such that $p$ is a quadratic residue modulo $q$.
\end{lemma}

The graphs we shall use are Cayley digraphs of the group $\PSL_2(q),$ and they are based on a wonderful construction of Lubotzky, Phillips, and Sarnak~\cite{LuPhiSar}. Fix the set-up of Theorem~\ref{thm:properdig}, and let $i$ be an integer with $i^2\equiv -1\pmod{q}$. In~\cite{LuPhiSar}, it is shown that $G=\PSL_2(q)=\SL_2(q)/Z$ has an inverse-closed set $S$ of generators (non-trivial cosets of $Z$) of the form $$Z\begin{bmatrix}
    a_0+ia_1&a_2+ia_3\\-a_2+ia_3&a_0-ia_1
\end{bmatrix}$$
 where $\mathrm{GF}(q)$ is identified with the integers modulo $q$ so that the $a_j$ are integers satisfying $a_0^2+a_1^2+a_2^2+a_3^2=p$ with $a_0>0$ and odd, and $a_j$ even for $j=1,2,3$. They show in~\cite{LuPhiSar} that there are exactly $p+1$ such quadruples $(a_0,a_2,a_2,a_3)$ and thus $\mathrm{Cay}(G,S)$ has valency $p+1$ and $q(q^2-1)/2$ vertices.
\begin{lemma}
    The set $S$ contains no involutions.
\end{lemma}
\begin{proof}
    Fix $a_0,a_1,a_2,$ and $a_3$ as described and let $x$ be the corresponding matrix. Suppose that $(Zx)^2=Z$. we shall show that $x\in Z$. First $$\begin{bmatrix}
    a_0+ia_1&a_2+ia_3\\-a_2+ia_3&a_0-ia_1
\end{bmatrix}^2=\begin{bmatrix}
    a_0^2-a_1^2+2ia_0a_1-a_2^2-a_3^2&2a_0(a_2+ia_3)\\
    2a_0(-a_2+ia_3)&a_0^2-a_1^2-2ia_0a_i-a_2^2-a_3^2
\end{bmatrix}.$$
The right hand side is scalar only if $2a_0(a_2+ia_3)=0=2a_0(-a_2+ia_3)$, but $a_0>0$ whence $ia_3=-a_2=-ia_3$, so $a_2=a_3=0.$ Additionally, since $(Zx)^2=Z$ we deduce that $a_0^2-a_1^2+2ia_0a_1=a_0^2-a_1^2-2ia_0a_1$, and thus $4ia_0a_1=0$. Therefore, since $q$ is odd and $a_0\ne 0$, it follows that $a_1=0$. Therefore, $$Zx=Z\begin{bmatrix}
    a_0&0\\0&a_0
\end{bmatrix}=Z,$$ 
that is, this is the trivial coset and hence does not lie in $S$. Thus $S$ contains no involutions.
\end{proof}

\begin{corollary}\label{cor:properdig}
    Set $p=4093$ and let $q\equiv1\pmod{4}$ be prime such that $p\ne q$ and $p$ is a quadratic residue modulo $q$. The group $G=\PSL_2(q)$ has a set $\widetilde{S}\subseteq S$ of $(p+1)/2$ generators. Let $\Gamma$ be the Cayley digraph $\mathrm{Cay}(G,\widetilde{S})$ and let $M$ be the adjacency matrix of $\Gamma$. Then
    \begin{enumerate}
        \item[(1)] $\Gamma$ has out-valency $(p+1)/2$; and
        \item[(2)] with $\hat{M}:=M+M^\top$ and $\lambda(\hat{M})$ as defined in Section~\ref{sec:stats}, $\lambda(\hat{M})\leq 2\sqrt{p}$.
    \end{enumerate}
\end{corollary}
\begin{proof}
    Since $S$ has no involutions, we have a partition $\widetilde{S}\sqcup\widetilde{S}^{-1}=S$. Since $S$ generates the finite group $G$, so too does $\widetilde{S}$ (and $\widetilde{S}^{-1})$. Claim (1) now follows as $$p+1=|S|=|\widetilde{S}|+|\widetilde{S}^{-1}|=2|\widetilde{S}|.$$
    For the second claim, observe that since $S$ has no involutions, $\hat{M}$ is the adjacency matrix of $\mathrm{Cay}(G,S)$. The result now follows from~\cite[Theorem 4.1]{LuPhiSar}, which bounds $\lambda(\hat{M}) \leq 2\sqrt p$.
\end{proof}
We now have enough information for the proof of Theorem~\ref{thm:properdig}.
\begin{proof}[Proof of Theorem~\ref{thm:properdig}]
    By~\cite[Table 1]{SunLiDing}, there exists a binary linear code $B$ of length $(p+1)/2=2047$ with dimension $1024$ and minimum distance at least 213. Let $q\equiv1\pmod{4}$ be any prime distinct from $p$ such that $p$ is a quadratic residue modulo $q$. Since $\Gamma=\mathrm{Cay}(G,\widetilde{S})$ (as defined in Corollary~\ref{cor:properdig}), a vertex-transitive digraph with out-valency $(p+1)/2$, we can define $C:=\D_\Gamma(G,G,1,B,B)$. 
    
    By Theorem~\ref{thm:ratebound}, $r(C)\geq 2r(B)-1=(2\cdot1024 -2047)/2047=1/2047=2/(p+1)$, as claimed. Similarly, by Theorem~\ref{thm:distanceboundD} and Corollary~\ref{cor:properdig}, $$\delta(C)\geq \frac{213}{p+1}\cdot\left(\frac{\frac{426}{p+1}-\frac{4\sqrt{p}}{p+1}}{2-\frac{4\sqrt{p}}{p+1}}\right)=\frac{213}{p+1}\cdot\left(\frac{213-2\sqrt{p}}{p+1-2\sqrt{p}}\right),$$ as desired.

    Finally, the fact that these codes form a good family follows from the fact that their parameters $r$ and $\delta$ are always bounded away from 0 (independent of $q$), and that this construction can be achieved by infinitely many $q$, by Lemma~\ref{lem:properdig}.
\end{proof}
%\section{Questions to explore/remaining tasks}
%\begin{itemize}
    %\item include example where nonzero $\overline{B}$ is properly contained in $B$.
    %\item notation:
    %$\mathfrak{g,h,t,s}$ for group elements, latin for codes, greek for vertices, $v$(?) for valency, mathcal for transversals latin for group name.
    %\item 
    %remove (weak) digraph code but comment on why we're not doing it.
 %   \item 
  %  define graph code first, show how generalises cayley codes, explain this is motivation.
   % \item flag final section which has examples with nice structure bipartite examples. maybe put bipartite thing in intro 
    %\item Petersen graph example (smallest non-cayley VT)
    %\item investigate what happens for SD when one of the input codes is well understood.
    %\item Can one find two non-isomorphic 3-valent graphs on same number of vertices where, given an input code not in $0,111,F^3$ get isomorphic nonzero graph codes $C$. Analogous question for digraphs but not doubled up arcs.{\color{red} Yes I think I have a pair on 20 vertices with the augmentation code. The groups are $D_{40}$ and $5\cn4$.}
    %\item do calculation checking if really subdirect product for the orbital decomposition. If not give example, if cant find then pose as open problem.
 %   \item worked example show ramanujan graphs have no involutions in gen set(matrix squares to scalar implies matrix is scalar ie is 1), use to show the digraph version is compatible with our construction. (p264 Lubotzky Phillips Sarnak) --- use code from the paper with m=11 (so $p=2^{12}-3$).
%\end{itemize}
\bibliographystyle{amsplain}
\bibliography{gcodebib}
\end{document}